\documentclass[11pt,letterpaper]{article}
\usepackage[utf8]{inputenc}

\usepackage{amsmath,amsthm,amssymb,amsfonts} \usepackage[boxruled,linesnumbered]{algorithm2e}

\usepackage{epsfig} \usepackage{latexsym,nicefrac,bbm}
\usepackage{xspace}
\usepackage{color,fancybox,graphicx,subfigure,fullpage}
\usepackage[top=1in, bottom=1in, left=1in, right=1in]{geometry}
\usepackage{tabularx} \usepackage{hyperref} 
\usepackage{pdfsync}
\usepackage{multicol}
\usepackage{enumitem}
\usepackage{multirow}
 \usepackage{url}

\usepackage{comment}

\usepackage{mathtools}

\usepackage{color}

\newcommand{\yesnum}{\addtocounter{equation}{1}\tag{\theequation}}

\makeatletter
\newcommand{\customlabel}[2]{%
\protected@write \@auxout {}{\string \newlabel {#1}{{#2}{\thepage}{#2}{#1}{}} }%
\hypertarget{#1}{}
}
\makeatother

\usepackage{mathabx}

\usepackage{amsmath,amssymb,bm}
\usepackage{mathtools}
\usepackage{etoolbox}
\usepackage{multirow}
\usepackage{multicol}
\usepackage[capitalise, nameinlink]{cleveref}

\makeatletter
\patchcmd{\@endtheorem}{\@endpefalse}{}{}{}
\patchcmd{\endproof}{\@endpefalse}{}{}{}
\makeatother

 \newtheorem{thm}{Theorem}[section]
 \newtheorem{lem}[thm]{Lemma}
 
 \newtheorem{cor}[thm]{Corollary}

 \newtheorem{defn}[thm]{Definition}
 
 \newtheorem{prob}[thm]{Problem}

\renewcommand{\bm}[1]{#1}

\renewcommand{\epsilon}{\varepsilon}

\newcommand{\leftbrace}{\left\lbrace}
\newcommand{\rightbrace}{\right\rbrace}

\def\eqref#1{equation~\ref{#1}}

\def\1{\bm{1}}

\def\eps{{\epsilon}}

\def\vlambda{{\bm{\lambda}}}
\def\vgamma{{\bm{\gamma}}}

\def\vu{{\bm{u}}}
\def\vv{{\bm{v}}}

\def\mA{{\bm{A}}}

\def\mH{{\bm{H}}}

\def\mM{{\bm{M}}}
\def\mN{{\bm{N}}}

\def\mU{{\bm{U}}}

\def\mLambda{{\bm{\Lambda}}}

\DeclareMathAlphabet{\mathsfit}{\encodingdefault}{\sfdefault}{m}{sl}
\SetMathAlphabet{\mathsfit}{bold}{\encodingdefault}{\sfdefault}{bx}{n}

\def\gD{{\mathcal{D}}}

\def\gH{{\mathcal{H}}}

\def\gM{{\mathcal{M}}}

\def\gO{{\mathcal{O}}}

\def\gR{{\mathcal{R}}}
\def\gS{{\mathcal{S}}}

\def\sS{{\mathbb{S}}}

\def\sZ{{\mathbb{Z}}}

\newcommand{\E}{\mathbb{E}}

\newcommand{\R}{\mathbb{R}}
\newcommand{\C}{\mathbb{C}}
\newcommand{\PR}{\mathbb{P}}

\DeclareMathOperator{\diag}{diag}
\DeclareMathOperator{\Tr}{Tr}

\DeclareMathOperator{\OPT}{OPT}

\newcommand{\inner}[2]{\left\langle#1,#2\right\rangle}
\DeclareMathOperator{\Lap}{Lap}

\title{\bf Private Matrix Approximation and  Geometry of Unitary Orbits}

\author{Oren Mangoubi \\ Worcester Polytechnic Institute \and Yikai Wu \\ Yale University \and Satyen Kale \\ Google Research \and Abhradeep Guha Thakurta \\ Google Research \and Nisheeth K. Vishnoi \\ Yale University}

\usepackage{times}

\begin{document}
 \date{} 
\maketitle

\begin{abstract}
 \noindent
Consider the following  optimization problem: Given $n \times n$  matrices   $A$ and $\Lambda$,  maximize $\langle A, U\Lambda U^*\rangle$ where $U$ varies over the unitary group $\mathrm{U}(n)$.
This problem seeks to approximate $A$ by a matrix whose spectrum is the same as $\Lambda$
and, by setting $\Lambda$ to be  appropriate diagonal matrices, one can recover  matrix approximation problems such as PCA and rank-$k$ approximation.
We study the problem of designing differentially private algorithms for this  optimization problem in settings where the matrix $A$ is constructed using users' private data. 
We give efficient and  private algorithms that come with upper and lower bounds on the approximation error.
Our results unify and improve upon several  prior works on  private matrix approximation problems.
They rely on extensions of  packing/covering number bounds for Grassmannians to unitary orbits which should be  of independent interest.
\end{abstract}

 \newpage
 \tableofcontents

  \newpage

\section{Introduction}

In machine learning and statistical data analysis, a widely used technique is to represent  data as a matrix and perform computations on the covariance matrix to extract statistical information from data.
For instance, consider the setting with $n$ users and where one represents the features of each user by a vector $x_i \in \mathbb{R}^d$, giving rise to the $d \times d$ covariance matrix $M=\sum_{i=1}^n x_ix_i^\top$.
In many applications, approximations to such matrices are sought to reduce the space/time required to perform computations or, to replace them by matrices with a specified spectrum \cite{sarwar2000analysis,paterek2007improving, koren2009matrix, beutel2015accams}.
An example of the first kind is the rank-$k$ approximation problem where one is given a positive integer $k$ and the goal is to find a rank-$k$ matrix $\mH$ which is ``close'' to $\mM$.
An example of the latter class of problems is the rank-$k$ PCA problem where one is given a positive integer $k$ and the goal is to output the matrix $H$ corresponding to the projection onto the subspace spanned by the top $k$ eigenvectors of $M$.
Closeness is usually measured using the spectral or Frobenius norm of $\mM-\mH$.
All of these problems are extensively studied and algorithms for these problems have been well studied and deployed; see \cite{BlumKannanHopcroft}.

Since such matrix approximation problems are often applied to matrices arising from user data (i.e. each user contributes one vector $x_i$ to the sum above), an important concern is to protect the privacy of the users.
Even without fixing a specific notion of privacy, traditional algorithms for these problems can leak information of users. 
For instance, suppose we know that a user vector $x$ is part of exactly one of the two covariance matrices $M$ and $M'$, but we cannot access the data matrices directly and can only obtain the information of the data matrices using PCA.
If we apply a traditional algorithm for rank-$k$ PCA onto $M$ and $M'$, we obtain  two projection matrices $H$ and $H'$ spanned by the top $k$ eigenvectors of $\mM$ and $\mM'$, respectively.
Then, if $x$ is in the subspace of $H$ but not in the subspace of $H'$, we know for sure that $x$ is part of $M$ but not of $M'$ -- 
leading to a breach of privacy of the data vector $x$.
Important examples of real-world privacy breaches in settings of this nature include the Netflix prize problem and \cite{bennett2007netflix} and 
 recommendation systems of Amazon and Hunch \cite{calandrino2011you}.
It is thus important to design private algorithms for fundamental matrix-approximation problems.

The notion of differential privacy has  arisen as an important formalization of what it means to protect privacy of individuals in a dataset \cite{dwork2006differential}. 
We say that two  Hermitian PSD  matrices $M$ and $M'$ are {\em neighbors} if each matrix is obtained from the other by replacing one user's vector by another user's vector. In other words, $M$ and $M'$ are neighbors
if and only if there exists $x, y \in \C^d$ such that $\|x\|_2, \|y\|_2\leq 1$ and $M' = M-x x^*+yy^*$. 
We can now define differentially private computations on matrices.
\begin{defn}[\textbf{Differential Privacy}]
\label{def:DP}
For a given $\eps \geq 0$ and $\delta \geq 0$, a randomized mechanism $\gM$ is said to be $(\eps, \delta)$-differentially private if for  any two neighboring matrices $M$ and $M'$  
and any measurable set of possible output $S$, it holds that
$$\PR[ \gM(M) \in S] \leq \exp(\epsilon) \cdot \PR[ \gM(M') \in S] + \delta.
$$
When $\delta = 0$, the mechanism is said to be $\eps$-differentially private.
\end{defn}

\noindent
There have been multiple works that give differentially private algorithms for matrix approximation problems, including rank-$k$ approximation \cite{kapralov2013differentially, upadhyay2016price,amin2019differentially} and rank-$k$ PCA \cite{chaudhuri2013near, dwork2014analyze, leake2020polynomial}.
Roughly, these algorithms can be divided into two categories:
Those satisfying pure differential privacy ($\eps$-differential privacy) \cite{chaudhuri2013near, kapralov2013differentially, amin2019differentially} and those satisfying $(\eps,\delta)$-differential privacy with a $\delta > 0$ \cite{chaudhuri2013near, dwork2014analyze, upadhyay2016price}.
Pure differential privacy provides better privacy protection and we focus on pure differential privacy in this paper.
All the algorithms mentioned above that come with pure differential privacy guarantees utilize the exponential mechanism \cite{mcsherry2007mechanism} (see \cref{thm:exponential}). 
This mechanism involves  sampling from an exponential distribution which, in turn, depends on the utility function chosen.
Among these algorithms, one of the algorithms by \cite{chaudhuri2013near} (PPCA) provides a near-optimal algorithm for PCA under pure differential privacy.
However, the error upper and lower bounds are only proved for the first principal component (the top  eigenvector).
In addition, their algorithm satisfying pure differential privacy (PPCA) is implemented with a Gibbs sampler which is not shown to  run in polynomial time.
\cite{kapralov2013differentially} provide two different algorithms under pure differential privacy for rank-$1$ approximation and rank-$k$ approximation.
They also provides error upper  and lower bounds for both problems.
Their rank-$1$ approximation relies on an efficient way to sample from a unit vector using the exponential mechanism. 
The algorithm outputs the sampled vector as the estimation of the first eigenvector of the input matrix.
The rank-$k$ approximation samples top $k$ eigenvectors iteratively.
The error bound is worse compared to the rank-$1$ case and there is a significant gap from the lower bound proved in the paper.

\cite{amin2019differentially} provide a differentially private algorithm for the version of rank-$k$ approximation problem when $k=d$ ({\em  covariance matrix estimation problem}). 
This problem is trivial without a privacy requirement: one can set the output $H$ as the input covariance matrix $M$. In the differentially private case, \cite{amin2019differentially} give an algorithm that samples eigenvectors iteratively using an exponential mechanism.
It uses a different error measure compared to \cite{kapralov2013differentially} and, hence, the error bounds cannot be compared directly.
However,  the algorithm of \cite{amin2019differentially}  only applies to the covariance matrix estimation problem which is a special case of the rank-$k$ approximation problem.

\section{Our Work}

\subsection{Unitary Orbit Optimization}

We first present a generalized problem that captures the matrix approximation problems mentioned above. 
The problem is a linear optimization problem over an orbit of the unitary group.
Recall that a matrix $U \in \C^{d \times d}$ is said to be unitary if $UU^* = I$.
The set of unitary matrices forms a group under matrix multiplication and is denoted by $\mathrm{U}(d)$.
$\mathrm{U}(d)$ is also a non-convex manifold.
For a given $d \times d$ Hermitian matrix $H$, $\mathrm{U}(d)$ acts on it by conjugation as follows: $H \mapsto UHU^*$ for a unitary matrix $U$.
Note that $H$ has the same eigenvalues as $UHU^*$ for any unitary matrix $U$.
Thus, the set of matrices obtainable from $H$ under this action have the same set of eigenvalues. 
Given a diagonal matrix $\mLambda \coloneqq  \diag\left(\lambda_1, \ldots, \lambda_d\right)$, we denote its unitary orbit:
\begin{equation*}
   \gO_{\mLambda} \coloneqq \lbrace \mU\mLambda\mU^*  : \mU \in \mathrm{U}(d)\rbrace.
\end{equation*}
\begin{prob}[\bf Unitary orbit optimization]
\label{prob:optim}
Given a   Hermitian  matrix $\mM \in \C^{d\times d}$ with eigenvalues $\gamma_1 \geq \cdots \geq \gamma_d$ and a list of  eigenvalues  $\lambda_1 \geq \cdots \geq \lambda_d$, 
the goal is to find a Hermitian matrix $\mH \in \gO_\mLambda$ with $\mLambda \coloneqq \diag\left(\lambda_1, \ldots, \lambda_d\right) \in \C^{d\times d}$ that maximizes $\inner{\mM}{\mH}\coloneqq \Tr(\mM^*\mH)$.
\end{prob}
Since $U\Lambda U^*$ has the same eigenvalues as $\Lambda$, this problem asks to find the ``closest'' matrix to a given matrix $M$, with eigenvalues identical to those of $\Lambda$.
This is a well-studied problem and the Schur-Horn Theorem implies that the optimal solution to this problem is the matrix $H = U\Lambda U^*$ where $U$ is a unitary matrix whose columns are the eigenvectors of $M$, attaining the optimal value $\sum_{i=1}^d \lambda_i \gamma_i$  \cite{schur1923,horn1954}.
Rank-$k$ PCA, rank-$k$ approximation, and covariance matrix estimation of a given covariance matrix $\mM$ can be reduced to Problem \ref{prob:optim} by a careful choice of $\lambda_i$'s. %
The {rank-$k$ PCA} problem 
is obtained by setting $\lambda_1 =\cdots = \lambda_k = 1$ and $\lambda_{k+1}=\cdots = \lambda_d=0$.
The rank-$k$ approximation problem is obtained by
setting $\lambda_i$ equal to the $i$-th largest eigenvalue of $\mM$ for $1 \leq i \leq k$, and to $0$ for $i > k$.
Finally, the covariance matrix estimation problem is obtained by setting $\lambda_i$ equal to the $i$-th largest eigenvalue of $\mM$ for all $i = 1, 2, \ldots, d$.
In the rank-$k$ approximation and covariance estimation problems, we consider both the setting where the eigenvalues $\lambda_1,\ldots,\lambda_k$ are given as prior ``non-private" information, as well as the more challenging setting when $\lambda_1,\ldots,\lambda_k$ are private.

\subsection{Upper Bound Results}

Our first result is an $\eps$-differentially private mechanism for Problem \ref{prob:optim}
when the matrix $\Lambda$ is non-private (as in the case of rank-$k$ PCA).
Our algorithm (Algorithm \ref{alg:dp_optim_HCIZ}) utilizes the exponential mechanism \cite{mcsherry2007mechanism} and samples $H$ from $\mathcal{O}_\Lambda$ from a  density that is close in infinity distance to   $\exp(\frac{\eps}{\lambda_1}\inner{\mM}{\mH})$.
However, to do this, we need a unitarily invariant  measure $\mu_\Lambda$ on $\mathcal{O}_\Lambda$.
Such a measure can be derived from the Haar measure on $\mathrm{U}(d)$;
see \cite{LeakeVishnoiSurvey}.
Note that while Problem \ref{prob:optim} makes sense for general Hermitian $M$, in our results we consider the case when $M$ is Hermitian and positive semidefinite (PSD) as in the case of a covariance matrix.

\begin{thm}[\textbf{Differentially private unitary orbit optimization}]
\label{thm:dp_optim_intro}
For any  $\eps \in (0,1)$, there is a randomized $\eps$-differentially private  algorithm (Algorithm \ref{alg:dp_optim_HCIZ}) such that given a $d \times d$ PSD  Hermitian matrix $\mM \in \C^{d\times d}$ with eigenvalues $\gamma_1 \geq \gamma_2 \geq \cdots \geq \gamma_d \geq 0$, the maximum rank of the output matrix $k\in[d]$, and a list of top $k$ nonnegative eigenvalues  of the output matrix $\lambda_1 \geq \lambda_2 \geq \cdots \geq \lambda_k \geq 0$, outputs a $d \times d$ PSD Herimitian matrix $\mH \in \C^{d\times d}$ with eigenvalues $\lambda_1, \lambda_2, \ldots, \lambda_k, 0, \ldots, 0$, where there are $d-k$ 0's. 
Moreover, for any $\beta \in (0,1)$, with probability at least $1-\beta$, we have 
$$\inner{\mM}{\mH} \geq \sum_{i=1}^k \gamma_i\lambda_i-  \Tilde{O}\left(\frac{dk\lambda_1}{\eps}\right),
$$
where $\Tilde{O}$ hides  logarithmic factors of $\frac{1}{\beta}$ and $\Tr(\mM)$.
The number of arithmetic operations required by this algorithm is polynomial in $\log \frac{1}{\epsilon}$, $\lambda_1$, $\gamma_1 -\gamma_d$, and the number of bits representing $\vlambda=(\lambda_1, \lambda_2, \ldots, \lambda_k)$ and $\vgamma=(\gamma_1, \gamma_2, \ldots, \gamma_d)$.
\end{thm}

\noindent 
This theorem is a generalization of the result in \cite{leake2020polynomial} which proved such a theorem for the special case of rank-$k$ PCA (when the orbit eigenvalues are  $\lambda_1=\cdots = \lambda_k=1$ and $\lambda_{k+1}=\cdots = \lambda_d=0$). 
Our algorithm leverages efficient algorithms 
 to sample approximately from such exponential densities on unitary orbits by
\cite{leake2020polynomial,MV21} that provide guarantees on the closeness of the target distribution and the actual distribution in infinity distance.\footnote{For two densities $\nu$ and $\pi$, the infinity distance is $\mathrm{d}_\infty(\nu, \pi)\coloneqq  \sup_{\theta} |\log  \frac{\nu(\theta)}{\pi(\theta)}|$.}
Note that the error bound in \cite{leake2020polynomial}  improves on \cite{kapralov2013differentially} but is weaker than in Theorem \ref{thm:dp_optim_intro} since Theorem \ref{thm:dp_optim_intro} holds with high probability while \cite{leake2020polynomial} only holds in expectation. 
The proof of Theorem \ref{thm:dp_optim_intro} appears in Section \ref{Appendix_upper_utility_bound} and uses a  covering number bound for the orbit $\mathcal{O}_\Lambda$ (Lemma \ref{lem:covering_number}) that generalizes the upper covering bounds for the Grassmannian \cite{szarek1982nets}.

Our next result considers the setting of Problem \ref{prob:optim} when $\Lambda$ is the spectrum of $M$:
$\lambda_i=\gamma_i$ for $1 \leq i \leq k$ and $\lambda_i=0$ for $i>k$ (as in the rank-$k$ covariance matrix approximation problem).
In this case, $\Lambda$ is also private and Algorithm \ref{alg:dp_optim_HCIZ} does not apply as such.
However, we show that adding Laplace noise to $\lambda_i$s, sorting them,  and then using Algorithm \ref{alg:dp_optim_HCIZ} suffices; see Algorithm \ref{alg:dp_rank_k}.

\begin{thm}[\textbf{Differentially private  rank-$k$ approximation}]
\label{thm:dp_rank_k}
Given a PSD Hermitian input matrix $\mM \in \gH_+^d$, a $k \in [d]$, and an $\eps > 0$.
Let the eigenvalues of $\mM$ be $\lambda_1 \geq \cdots \geq \lambda_d \geq 0$.
There exists a randomized $\eps$-differentially private algorithm (Algorithm \ref{alg:dp_rank_k}), which outputs a rank-$k$ matrix $\mH \in \gH_+^d$ and a list of estimated eigenvalues $\Tilde{\lambda}_1, \ldots, \Tilde{\lambda}_k$.
For any $\beta \in (0,1)$, with probability at least $1-\beta$, for all $i \in [k]$, we have 
$
    |\Tilde{\lambda}_i - \lambda_i| \leq O\left(\frac{1}{\eps}\log \frac{1}{\beta}\right),
$ and   $$\|\mM-\mH\|_F^2 \leq
 \sum_{\ell=k+1}^d\lambda_\ell^2 +  \Tilde{O}\left(\frac{k}{\epsilon^2}+\frac{dk}{\eps}\left(\lambda_1 + \frac{1}{\eps}\right)\right),$$
where $\Tilde{O}$ hides  logarithmic factors of $ \frac{1}{\beta}$ and $\sum_{\ell=1}^k \lambda_\ell$.
The number of arithmetic operations required by this algorithm is polynomial in $\log \frac{1}{\epsilon}$, $\lambda_1$, and the number of bits representing $\vlambda=(\lambda_1, \lambda_2, \ldots, \lambda_d)$.
\end{thm}
\noindent The proof of Theorem \ref{thm:dp_rank_k} is an extension of the proof of Theorem \ref{thm:dp_optim_intro} and appears in Section \ref{Appendix_covariance_estimation}.
Since the (full-rank) covariance matrix estimation problem is a special case of the rank-$k$ approximation problem (when we set $k=d$), the above result immediately applies in this case.
Theorem \ref{thm:dp_rank_k} improves upon the bound  in \cite{amin2019differentially}: 
 Roughly, when the covariance matrix has its largest eigenvalue within a constant factor of its middle eigenvalue $\lambda_1 = O(\lambda_{\frac{k}{2}})$, our bound of $\Tilde{O}(\frac{dk}{\eps}(\lambda_1 + \frac{1}{\eps}))$ is $O(d)$ better than the bound $\Tilde{\Omega}(\frac{kd^2}{\eps})$ of \cite{amin2019differentially}. 
This includes the setting when the input matrix $M$ is a random sample covariance matrix from the Wishart distribution \cite{wishart1928generalised} (that is $M = \frac{1}{m} X^\top X$, where $X$ is a $d\times m$ matrix with i.i.d. standard Gaussian entries), as such a matrix has, with high probability, $\lambda_1 = O(\lambda_{\frac{k}{2}})$ for any $m,d$, where $k= \min(m,d)$. 
We discuss these examples in detail in Section \ref{sec_examples}. 

Note that in Theorem \ref{thm:dp_rank_k} we do not lose utility due to privatization of the eigenvalues whenever $\lambda_1 \geq \Omega(\frac{1}{\epsilon})$, which is often the case in practice. 
In this case the utility bound $\tilde{O}(\frac{dk\lambda_1}{\eps})$ in Theorem \ref{thm:dp_optim_intro}  (where we assume the eigenvalues are “public”, and do not have an eigenvalue privatization step) is the same as the utility bound $\tilde{O}( \frac{dk}{\eps}(\lambda_1 + \frac{1}{\eps}))$  in  Theorem 2.3 (where we {\em do} privatize  eigenvalues).

\subsection{Lower Bound Results}
We give a lower bound for an $\epsilon$-differentially private algorithm in the  case where the eigenvalues $\gamma_1, \ldots, \gamma_d$ of the input matrix are equal to the eigenvalues $\lambda_1, \ldots, \lambda_d$ of the output matrix.
Note that this lower bound holds even when the eigenvalues of the input matrix are given to the algorithm as prior non-private information.

\begin{thm}[\bf Error lower bound] \label{thm_lower_bound_general_orbit}  
Suppose that $\lambda_1 \geq \cdots \geq \lambda_d \geq 0$ and $\epsilon>0$. 
Then for any $\epsilon$-differentially private algorithm $\mathcal{A}$ which takes as input a Hermitian matrix and outputs a rank-$k$ Hermitian matrix with eigenvalues $\lambda_1, \ldots, \lambda_k$,  there exists a $d \times d$ PSD Hermitian matrix $M$ with eigenvalues $\gamma_i = \lambda_i$, $i \in [d]$, such that, with probability at least $\frac{1}{2}$, the output  $H:=\mathcal{A}(M)$ of the algorithm satisfies
\begin{equation} \label{eq_bad_utility2}
    \| M - H \|_F^2 \geq \Omega\left(\sum_{\ell=k+1}^d \lambda_\ell^2 +  \frac{ d}{{\max(\lambda_1 \sqrt{\eps}, \sqrt{d}})^2}
\max_{1\leq i \leq \frac{d}{2}}   i \times (\lambda_i -\lambda_{d-i +1})^2\right).
\end{equation}
\end{thm}

\noindent 
Note that the r.h.s. of our lower bound is never larger than $\Omega(\sum_{\ell=1}^d \lambda_\ell^2)$; this is true for any error lower bound since the diameter $D_{\|\cdot\|_F}(\mathcal{O}_\Lambda)$ of the unitary orbit is $D_{\|\cdot\|_F}(\mathcal{O}_\Lambda) := \sup_{M, H \in \mathcal{O}_\Lambda} \|M-H\|_F = O(\sqrt{\sum_{\ell=1}^d \lambda_\ell^2})$.
The proof of Theorem \ref{thm_lower_bound_general_orbit} is given in Section \ref{appendix_lower_utility_bound}.
   The proof of Theorem \ref{thm_lower_bound_general_orbit} relies on a novel packing number lower bound for the unitary orbit $\mathcal{O}_\Lambda$ (Theorem \ref{lemma_orbit_packing_improved}).
   As a first attempt we show a packing number bound for the entire unitary orbit (Inequality \ref{eqref_packing_bound_entire_orbit}).
   Unfortunately the resulting utility error lower bound (Inequality \ref{eq_lower_overview_4} in the proof overview) is (roughly) proportional to $e^{-\frac{1}{d} D_{\|\cdot\|_F}(\mathcal{O}_\Lambda)^2}$,  
   which is exponentially small in the eigenvalues $\lambda_1,\ldots, \lambda_d$.
  To achieve an error     bound polynomial in the $\lambda$'s, we instead show a packing bound on a ball of radius $\omega$ inside the orbit, where $\omega$ is carefully chosen to ensure that the error bound is polynomial in $\lambda_1,\ldots, \lambda_d$.

Next, we give a corollary of Theorem \ref{thm_lower_bound_general_orbit}, which provides a lower bound for the rank-$k$ approximation problem (which includes the covariance matrix estimation problem as a special case).

\begin{cor}[\bf Lower bound for covariance estimation] \label{cor_lower_covariance_estimation}  
Suppose that $\lambda_1 \geq \cdots \geq \lambda_d \geq 0$ and $\epsilon>0$.
 Then for any $\epsilon$-differentially private algorithm $\mathcal{A}$ which takes as input a Hermitian matrix $M$ and outputs a rank-$k$ Hermitian matrix $H= \mathcal{A}(M)$,  there exists a $d \times d$ PSD Hermitian matrix $M$ with eigenvalues $\lambda_i$, $i \in [d]$, such that, with probability at least $\frac{1}{2}$, the output  $H:=\mathcal{A}(M)$ of the algorithm satisfies
\begin{equation*}      \| M - H \|_F^2   \geq  \Omega \left(\sum_{\ell=k+1}^d \lambda_\ell^2 + \frac{ d}{{\max(\lambda_1 \sqrt{\eps}, \sqrt{d})^2}} 
\max_{1\leq i \leq \frac{d}{2}}   i \times (\lambda_i -\lambda_{d-i +1})^2\right).
\end{equation*}

\end{cor}

\noindent The proof of  Corollary \ref{cor_lower_covariance_estimation}  is given in Section \ref{appendix_lower_utility_bound}.
Note that, unlike in Theorem \ref{thm_lower_bound_general_orbit}, the output matrix in Corollary \ref{cor_lower_covariance_estimation} is allowed to be any matrix, and need not have the same eigenvalues as the input matrix.
To verify that the lower bounds in Theorem \ref{thm_lower_bound_general_orbit} and Corollary \ref{cor_lower_covariance_estimation} are indeed lower than the upper bounds in Theorems  \ref{thm:dp_optim_intro} and \ref{thm:dp_rank_k}, we observe that when the input matrix is of any rank $1 \leq k \leq d$, $\max_{1\leq i \leq d}   i \times (\lambda_i -\lambda_{d-i +1})^2 \leq k \times \lambda_1^2$.
Thus, the r.h.s. of the lower bound in Theorem \ref{thm_lower_bound_general_orbit} and Corollary \ref{cor_lower_covariance_estimation} is at most $\frac{d k}{\eps}$, up to a constant factor.
On the other hand the upper bounds in Theorems \ref{thm:dp_optim_intro} and \ref{thm:dp_rank_k} are each at least as large as $\frac{dk}{\eps}(\lambda_1 + \frac{1}{\eps})$, which is greater than $\frac{d k}{\eps}$.

When the input matrix is rank-$k$, and $\lambda_{\frac{k}{4}} - \lambda_{\frac{3k}{4}} = \Omega(\lambda_1)$, Corollary \ref{cor_lower_covariance_estimation} implies that, with probability at least $\frac{1}{2}$, $\|M - \mathcal{A}(M)\|_F^2 \geq \frac{dk}{\eps}$ if    $\lambda_1 \geq \Omega(\frac{\sqrt{d}}{\eps})$ and $\|M - \mathcal{A}(M)\|_F^2 \geq k \lambda_1^2$ if $\lambda_1 \leq \Omega(\frac{\sqrt{d}}{\eps})$. 
Thus, our lower bound matches our upper bound from  Theorem \ref{thm:dp_rank_k} up to a factor of  $\frac{\lambda_1}{\eps}$ if $\lambda_1 \geq \Omega(\frac{\sqrt{d}}{\eps})$ and a factor of  $\frac{d}{\lambda_1 \eps}$ otherwise.
This includes the setting when the input matrix $M$ is a random sample covariance matrix from the Wishart distribution \cite{wishart1928generalised}, as such a matrix has, with high probability, $\lambda_1 = O(\lambda_{\frac{k}{2}})$ and $\lambda_{\frac{k}{4}} - \lambda_{\frac{3k}{4}} = \Omega(\lambda_1)$ for any $m,d$, where $k= \min(m,d)$.

The only previous lower bound we are aware of for the problem of (pure) differentially private rank-$k$ covariance matrix estimation is from \cite{kapralov2013differentially}.
Roughly, their result says that if, for any $\omega>0$ and $\lambda_1>0$, we have $\lambda_1 > \frac{1}{\eps} k (d-k) \log(\frac{1}{\omega})$, then, for any $\eps$-differentially private algorithm $\mathcal{A}$, there exists a matrix $M$ with top eigenvalue $\lambda_1$, such that the error (measured in the spectral norm) has a lower bound of $\|M - \mathcal{A}(M)\|_2 \geq \lambda_{k+1} + \delta \lambda_1$ with positive probability, where $\lambda_{k+1}$ is the $k+1$'st eigenvalue of the matrix $M$ guaranteed by their result.
Since only a condition on the top eigenvalue $\lambda_1$ is specified in their result, to show their result it is sufficient to produce an input matrix $M$ satisfying their lower bound with $\lambda_{k+1} = 0$, that is, an input matrix of rank $k$, and this is what they show in their proof.
Solving for the value of $\omega$ which maximizes their lower bound, one gets that their lower bound implies $\|M- \mathcal{A}(M)\|_2 - \lambda_{k+1} \geq \Omega(e^{-\frac{\lambda_1\eps}{k(d-k)}} \lambda_1)$.

While our lower bound is stated in terms of Frobenius norm, to see what our results give for the spectral norm error, we can use the fact that the Frobenius norm distance between two rank-$k$ projection matrices is at most $O(\sqrt{k})$ times the spectral norm distance to obtain a spectral norm bound.
In the case where the input matrix is rank-$k$, our result implies a error bound of $\|M - H\|_2 \geq \Omega(\frac{d}{\eps})$ if, e.g., $\lambda_1 \eps > \Omega(\sqrt{d})$.
Thus, our lower bound is larger by a factor of roughly $\frac{d}{\eps \lambda_1}e^{\frac{\lambda_1\eps}{k(d-k)}}$.

For the general unitary orbit approximation problem, Theorem \ref{thm_lower_bound_general_orbit} implies the following utility lower bound on the Frobenius norm utility.
\begin{cor}[\textbf{Lower bound for general $\gamma$ and $\lambda$}] \label{cor_lower_general_gamma_lambda}
Suppose that $\gamma_1 \geq \cdots \geq \gamma_d \geq 0$, and $\lambda_1 \geq \cdots \geq \lambda_d \geq 0$,  and $\epsilon>0$. 
Then for any $\epsilon$-differentially private algorithm $\mathcal{A}$ which takes as input a Hermitian matrix and outputs a rank-$k$ Hermitian matrix with eigenvalues $\lambda_1, \ldots, \lambda_k$,  there exists a $d \times d$ PSD Hermitian matrix $M$ with eigenvalues $\gamma_i$, $i \in [d]$, such that, with probability at least $\frac{1}{2}$, the output  $H:=\mathcal{A}(M)$ of the algorithm satisfies
\begin{equation*}     \| M - H \|_F^2 \geq \Omega\left(\sum_{\ell=k+1}^d \gamma_\ell^2 +  \frac{ d}{{\max(\gamma_1 \sqrt{\eps}, \sqrt{d}})^2}
\max_{1\leq i \leq \frac{d}{2}}   i \times (\gamma_i -\gamma_{d-i +1})^2\right).
\end{equation*}
\end{cor}
The proof of  Corollary \ref{cor_lower_general_gamma_lambda}  is given in Section \ref{appendix_lower_utility_bound}.
 Corollary \ref{cor_lower_general_gamma_lambda}  says that, given any $\gamma_1 \geq \cdots \geq \gamma_d \geq 0$ and any $\lambda_1 \geq \cdots \geq \lambda_d \geq 0$, the same lower utility bound given in Theorem \ref{thm_lower_bound_general_orbit} holds (with $\gamma_i$ taking the place of $\lambda_i$ on the r.h.s. of the inequality) even if the eigenvalues $\lambda_i$ of the output matrix are not equal to $\gamma_i$.

\subsection{Packing Number Bounds for Unitary Orbits}
As our main technical tool for proving the lower bounds on the error in Theorem \ref{thm_lower_bound_general_orbit} and Corollary \ref{cor_lower_covariance_estimation}, we will show packing number bounds for the unitary orbit.
For any set $S$ in a normed vector space with norm $\vvvert \cdot \vvvert$ and any $\zeta>0$, we define a $\zeta$-packing of the set $S$ with respect to $\vvvert \cdot \vvvert$ to be any collection of points $\{z_1, \ldots, z_J\} \subseteq S$, where $J \in \mathbb{N}$, such that $\vvvert z_s - z_t \vvvert \geq \zeta$ for any $s, t \in [J]$.
We define the packing number $P(S, \vvvert \cdot \vvvert, \zeta)$ to be the supremum of the number of points in any $\zeta$-packing of $S$.
We also denote by $B(X,r):= \{Z \in \mathbb{C}^{d\times d}  : \vvvert Z-X\vvvert \leq r \}$ a ball of radius $r$ with center $X$ with respect to the norm $\vvvert \cdot \vvvert$. 
We show the following lower bound on the packing number of any unitary orbit $\mathcal{O}_\Lambda$ with respect to the Frobenius norm $\|\cdot \|_F$, and also provide a bound on the packing number of any ball $B \cap \mathcal{O}_\Lambda$ which is a subset of the unitary orbit.
Since the $\zeta$-packing and $\zeta$-covering numbers of any set are equal up to a factor of 2 in $\zeta$  (see \eqref{eq_cover_pack}), our packing number lower bound also implies a lower bound on the covering number of the unitary orbit.

\begin{thm}[\bf Packing number lower bound for unitary orbits] \label{lemma_orbit_packing_improved}
There exist universal constants $C>c>0$ such that, for any $\Lambda = \mathrm{diag}(\lambda_1, \ldots, \lambda_d)$, and any $\omega, \zeta> 0$, and any $X \in \mathcal{O}_\Lambda$,
\begin{align} \label{eqref_packing_bound_subset}
  \log P(B(&X,\omega)\cap \mathcal{O}_\Lambda, \| \cdot\|_F , \zeta) \nonumber\\ &\geq \max_{1\leq i < j \leq d}  2i\times (d-j+1) \times  \log  \left(\frac{\min(\omega, \lambda_1\sqrt{i}, \lambda_1\sqrt{d-j+1}) \times (\lambda_i - \lambda_j)}{ 2C\lambda_1 \zeta}\right).
 \end{align}
Moreover, we get the following bound for the packing number of the entire unitary orbit $\mathcal{O}_\Lambda$:
\begin{align} \label{eqref_packing_bound_entire_orbit}
    \log P(\mathcal{O}_\Lambda, &\| \cdot\|_F , \zeta) \nonumber\\
    &\geq \max_{1\leq i < j \leq d}  2 i\times (d-j+1) \times \log \left(\frac{c 
   \min(\sqrt{i}, \sqrt{d-j+1})
    \times(\lambda_i - \lambda_{j})}{\zeta}\right).
\end{align}
\end{thm}
The proof of Theorem \ref{lemma_orbit_packing_improved} is given in Section  \ref{Appendix_packing_number}.
The bound in Theorem \ref{lemma_orbit_packing_improved} depends on the gaps $\lambda_i - \lambda_j$ between the eigenvalues of $\Lambda$, and is largest when there is a large gap between eigenvalues $\lambda_i - \lambda_j$ such that both  $i$ and $d-j$ are large.
A special case of the unitary orbits is the Grassmannian manifold $\mathcal{G}_{d,k}$ for any $k\leq d$, which is the set of $k$-dimensional subspaces in a $d$-dimensional vector space.
Identifying each subspace $\mathcal{V}\in \mathcal{G}_{d,k}$ with its projection matrix, the Grassmannian $\mathcal{G}_{d,k}$ has a one-to-one correspondence with the unitary orbit $\lambda_1 = \cdots =\lambda_k = 1$ and $\lambda_{k+1} = \cdots = \lambda_d=0$, and any norm on the unitary orbit induces a norm on the Grassmannian.
Theorem \ref{lemma_orbit_packing_improved} generalizes the covering/packing number {\em lower} bounds for the (complex) Grassmannian of \cite{szarek1982nets, szarek1998metric}  (restated as Lemma \ref{lemma_covering_Grassmanian} in the Section; see also e.g. \cite{pajor1998metric} and \cite{kapralov2013differentially} for different proofs of the same result), to a lower bound on the covering/packing number of any unitary orbit $\mathcal{O}_\Lambda$.  
Namely, in the special case where $\lambda_1 = \cdots =\lambda_k = 1$ and $\lambda_{k+1} = \cdots = \lambda_d=0$, the r.h.s. of Theorem \ref{lemma_orbit_packing_improved} is just  $2d\times (d-k) \log(\frac{c D_{\| \cdot\|_F}(\mathcal{G}_{d,\, \,k})}{\zeta})$,
since the diameter of the Grassmannian is
 $D_{\| \cdot\|_F}(\mathcal{G}_{d,\, \,k}) = c' \min(\sqrt{k}, \sqrt{d-k+1})$, for universal constant $c'$.

\section{Comparison of Our Bounds in Different Examples} \label{sec_examples}
In this section we compare our upper bound and lower bound theorems to key prior works. 
In our notation, the main result of \cite{amin2019differentially} can be written as follows.
\begin{thm}[\textbf{\cite{amin2019differentially}}]
\label{thm:iterative_eigenvector_sampling_utility}
Given a PSD symmetric covariance matrix $\mM \in  \R^{d\times d}$, a privacy budget $\eps > 0$, and privacy parameters $\eps_0, \eps_1, \ldots, \eps_d$, where $\sum_{i=0}^d \eps_i = \eps$.
Let the eigenvalues of $\mM$ be $\lambda_1 \geq \cdots \geq \lambda_d$.
There is a polynomial time algorithm that outputs a matrix $\mH \in \gS_+^d$ such that 
for any $\beta \in (0,1)$, with probability at least $1-\beta$, 
$
    \|\mM-\mH\|_F^2 \leq \Tilde{O}\left(\frac{d}{\epsilon_0^2}+\sum_{i=1}^d \frac{d}{\epsilon_i}\lambda_i\right)
$,
where $\Tilde{O}$ hides the logarithmic factors of $\frac{1}{\beta}$, $d$, and $\lambda_i$'s.
\end{thm}
\noindent
In the following we provide comparisons to our results for problems where the output matrix has (nearly) the same eigenvalues as the input matrix.
For simplicity, we denote  by $\lambda_1, \ldots, \lambda_d$ the eigenvalues of the {input} matrix, and by $\tilde{\lambda}_i$ the (privatized) eigenvalues of the output matrix.

\paragraph{Projection matrices.}
We first consider the  case when the input matrix $M$ is a scalar multiple of a projection matrix of some rank $k>0$.
In this case, the first $k$ eigenvalues of the input matrix all have the same value as the top eigenvalue $\lambda_1$, and the remaining $d-k$ eigenvalues are all $0$.

\noindent
{\em Upper bound (Theorem \ref{thm:dp_rank_k}):}
 When the input matrix is a rank-$k$ projection matrix, Theorem \ref{thm:dp_rank_k} gives a bound of
 $\|\mM-\mH\|_F^2 \leq \Tilde{O}\left(\frac{dk}{\eps^2}\right)$ with probability at least $1-\beta$, 
where $\Tilde{O}$ hides  logarithmic factors of $ \frac{1}{\beta}$ and $k$.
When the input matrix is a scalar multiple of a rank-$k$ projection matrix with top eigenvalue $\lambda_1$, Theorem \ref{thm:dp_rank_k} gives a bound of
 $\|\mM-\mH\|_F^2 \leq \Tilde{O}\left(\frac{dk}{\eps}(\lambda_1 + \frac{1}{\eps})\right)$ with probability at least $1-\beta$.

 \noindent
 {\em Upper bound in \cite{amin2019differentially} (Theorem \ref{thm:iterative_eigenvector_sampling_utility}): }
    When the input matrix is a rank-$k$ projection matrix, the error bound in Theorem \ref{thm:iterative_eigenvector_sampling_utility} is just $
    \|\mM-\mH\|_F^2 \leq E'$, where $E' = \Tilde{O}\left(\frac{d}{\epsilon_0^2}+\sum_{i=1}^k \frac{d}{\epsilon_i}\right)$.
This bound is minimized (up to a constant factor) by setting the privacy budget to be $\eps_0 = \frac{\eps}{2}$ and $\eps_i = \frac{\eps}{2k}$ for each $i \geq 1$.
Hence, $E' \geq \Tilde{\Omega}(\frac{kd^2}{\eps})$.

More generally, when the input matrix is a scalar multiple of a rank-$k$ projection matrix with top eigenvalue $\lambda_1$, the upper bound $E'$ in Theorem \ref{thm:iterative_eigenvector_sampling_utility}  has $E' \geq \Tilde{\Omega}(\frac{kd^2 \lambda_1}{\eps})$.
Thus, our bound in Theorem \ref{thm:dp_rank_k} is smaller than the bound $E'$ of Theorem \ref{thm:iterative_eigenvector_sampling_utility} by a factor of $\tilde{O}(d)$.

\noindent
{\em Lower bounds (Theorem \ref{thm_lower_bound_general_orbit} and Corollary \ref{cor_lower_covariance_estimation}}):
When the input matrix is a scalar multiple of a rank-$k$ projection matrix, our error lower bound is   $\| M - H \|_F^2 \geq \Omega(\min(\frac{k(d-k)}{\eps}, \lambda_1^2 k))$.
Note that the r.h.s. of the lower bound cannot be greater than $k \lambda_1^2$, since $\sup_{M,H \in \mathcal{O}_\Lambda}\|M - H\|_F^2 = O(k \lambda_1^2)$ if $\Lambda$ is rank-$k$.
In this case our lower bound matches our upper bound up to a factor of  $\frac{\lambda_1}{\eps}$ if $\lambda_1 \geq \Omega(\sqrt{d})$ and a factor of  $\frac{d}{\lambda_1 \eps}$ if $\lambda_1 \leq O(\sqrt{d})$.
While Corollary \ref{cor_lower_covariance_estimation} is stated in terms of the Frobenius norm, we can also get a bound for error defined in the spectral norm by using the fact that $\|M- \mathcal{A}(M)\|_F \leq \sqrt{2\min(k, d-k)} \|M- \mathcal{A}(M)\|_2$ since $M$ and $\mathcal{A}(M)$ are rank-$k$ matrices.
Thus our Corollary \ref{cor_lower_covariance_estimation} also implies a lower bound of $\|M- \mathcal{A}(M)\|_2 \geq \Omega(\frac{d}{\eps})$ with probability at least $\frac{1}{2}$ when the input matrix is a scalar multiple of a rank-$k$ projection matrix.
 In comparison, the lower bound from \cite{kapralov2013differentially}, which also considers the setting where the input and output of the algorithm are (scalar multiples of) rank-$k$ projection matrices, is $\|M- \mathcal{A}(M)\|_2 \geq \Omega(e^{-\frac{\lambda_1\eps}{k(d-k)}} \lambda_1)$.
 Thus, our lower bound is larger by a factor of $\frac{d}{\eps \lambda_1}e^{\frac{\lambda_1\eps}{k(d-k)}}$.

\paragraph{Matrices with condition number $O(1)$ and large eigenvalue gaps.}
We consider the case where the eigenvalues $\lambda_1,\ldots, \lambda_d$ of the input matrix $M$ are such that the input matrix has rank $k$ with condition number $\frac{\lambda_1}{\lambda_k} = O(1)$ (and more generally when we may only have $\frac{\lambda_1}{\lambda_{\frac{k}{2}}} = O(1)$) and also has a gap in the eigenvalues of $\lambda_{\frac{k}{4}} - \lambda_{\frac{3k}{4}} = \Omega(\lambda_1)$.

 \noindent
  {\em Upper bound (Theorem \ref{thm:dp_rank_k}):}
 Theorem \ref{thm:dp_rank_k} gives a bound of
 $\|\mM-\mH\|_F^2 \leq \Tilde{O}\left(\frac{dk}{\eps}\left(\lambda_1 + \frac{1}{\eps}\right)\right)$ with probability at least $1-\beta$, 
where $\Tilde{O}$ hides  logarithmic factors of $ \frac{1}{\beta}$, $\lambda_1$, and $k$.

\noindent
  {\em Upper bound in \cite{amin2019differentially} (Theorem \ref{thm:iterative_eigenvector_sampling_utility}):}
By letting $\eps_0 = O(\eps)$ in Theorem \ref{thm:iterative_eigenvector_sampling_utility}, the term in the error due to eigenvalue approximation is the same as $\Tilde{O}\left(\frac{d}{\epsilon_0^2}\right)$ for both algorithms and can thus be ignored.
The remaining term in the bound in Theorem \ref{thm:dp_rank_k} is $E\coloneqq \Tilde{O}\left(\frac{d k\lambda_1}{\eps}\right)$ and that in Theorem \ref{thm:iterative_eigenvector_sampling_utility}  is $E' \coloneqq \Tilde{O}\left(\sum_{i=1}^d \frac{d}{\epsilon_i}\lambda_i\right)$ which, in turn, depends on how the total privacy budget $\eps$ is distributed among the $\eps_i$s.
With probability $1-\beta$, $\Tilde{\lambda}_1 \approx  O(\lambda_1+\frac{1}{\epsilon}\log \frac{1}{\beta})$. 
Thus, when $\lambda_1 \geq \frac{1}{\epsilon}\log \frac{1}{\beta}$,   $\Tilde{\lambda}_1 = \Theta(\lambda_1)$.
In this case, $E = \Tilde{O}(\frac{d k\lambda_1}{\epsilon})$.
When $\frac{\lambda_1}{\lambda_k} = O(1)$ (or even if we just have the weaker condition that  $\frac{\lambda_1}{\lambda_{\frac{k}{2}}} = O(1)$), $ E' \geq \Tilde{\Omega}(d\sum_{i=1}^{\frac{d}{2}}\frac{1}{\epsilon_i} \lambda_i) \geq \Tilde{\Omega}(d\lambda_1\sum_{i=1}^{\frac{d}{2}}\frac{1}{\epsilon_i})$.
 Since $\sum_{i=1}^d\epsilon_i = \eps$, the quantity $\sum_{i=1}^{\frac{k}{2}}\frac{1}{\epsilon_i}$ $E'$ is minimized when $\epsilon_i \coloneqq O(\frac{\epsilon}{k})$ for each $i \leq \frac{k}{2}$ and  $\epsilon_i = 0$ for $i > \frac{k}{2}$. 
Hence, $E' \geq \Tilde{\Omega}(\frac{d^2 k\lambda_1}{\eps})$.
Thus, in this case, the bound $E$ from our Theorem \ref{thm:dp_rank_k} is $\Tilde{O}(d)$ smaller than the bound $E'$ from Theorem \ref{thm:iterative_eigenvector_sampling_utility}.

\noindent
{\em Lower bounds (Theorem \ref{thm_lower_bound_general_orbit} and Corollary \ref{cor_lower_covariance_estimation}):}
If the input matrix is rank-$k$ ($\lambda_i = 0$ for $i >k$), and $\lambda_{\frac{k}{4}} - \lambda_{\frac{3k}{4}} = \Omega(\lambda_1)$, then the bound in Corollary \ref{cor_lower_covariance_estimation} implies that $\|M - \mathcal{A}(M)\|_F^2 \geq   \Omega(\min(\frac{dk}{\eps}, k \lambda_1^2))$ with probability at least $\frac{1}{2}$.
Thus, our lower bound matches our upper bound from  Theorem \ref{thm:dp_rank_k} up to a factor of  $\frac{\lambda_1}{\eps}$ if $\lambda_1 \geq \Omega(\sqrt{d})$ and a factor of  $\frac{d}{\lambda_1 \eps}$ if $\lambda_1 \leq O(\sqrt{d})$.

\paragraph{Wishart random matrices.}
We consider the setting where the input matrix $M$ is a random sample covariance matrix from the Wishart distribution \cite{wishart1928generalised} (that is $M = \frac{1}{d} X^\top X$, where $X$ is an $m\times d$ matrix with i.i.d. standard Gaussian entries).
As in the previous examples, we denote by $\lambda_1,\ldots, \lambda_d$ the eigenvalues of the {\em input} matrix.
    \noindent
    {\em Upper bound (Theorem \ref{thm:dp_rank_k}):}
Theorem \ref{thm:dp_rank_k} gives a bound of
 $\|\mM-\mH\|_F^2 \leq \Tilde{O}\left(\frac{dk}{\eps}\left(\lambda_1 + \frac{1}{\eps}\right)\right)$ with probability at least $\frac{1}{2}$
where $\Tilde{O}$ hides  logarithmic factors of $ \frac{1}{\beta}$, $\lambda_1$, and $k$.

\noindent    
    {\em Upper bound in \cite{amin2019differentially} (Theorem \ref{thm:iterative_eigenvector_sampling_utility}):}
From concentration results for random matrices, we have, with high probability, that $\lambda_1 = O(\lambda_{\frac{k}{2}})$ for any $m,d$, where $k= \min(m,d)$ is the rank of $M$. 
From the discussion in the previous section we have that, whenever $\lambda_1 = O(\lambda_{\frac{k}{2}})$, the bound $E'$ of Theorem \ref{thm:iterative_eigenvector_sampling_utility} on the error $\|\mM-\mH\|_F^2$ satisfies $E' \geq \Tilde{\Omega}(\frac{d^3\lambda_1}{\eps})$.
Thus, if the input matrix is a Wishart random matrix, with high probability, the bound given in our Theorem \ref{thm:dp_rank_k} is $\Tilde{O}(d)$ smaller than the bound $E'$.

\noindent
{\em Lower bound (Theorem \ref{thm_lower_bound_general_orbit} and Corollary \ref{cor_lower_covariance_estimation}):}
From concentration results for random matrices, we also have that, with high probability, there is a large eigenvalue gap $\lambda_{\frac{k}{4}} - \lambda_{\frac{3k}{4}} = \Omega(\lambda_1)$ for any $m,d$, where $k= \min(m,d)$ is the rank of $M$.
Thus, from the discussion in the previous section, the bound in Corollary \ref{cor_lower_covariance_estimation} implies that $\|M - \mathcal{A}(M)\|_F^2 \geq  \Omega(\min(\frac{dk}{\eps}, k \lambda_1^2))$ with probability at least $\frac{1}{2}$.
Thus, our lower bound matches the upper bound from  Theorem \ref{thm:dp_rank_k} up to a factor of $\frac{\lambda_1}{\eps}$ if $\lambda_1 \geq \Omega(\sqrt{d})$ and a factor of  $\frac{d}{\lambda_1 \eps}$ if $\lambda_1 \leq O(\sqrt{d})$.

\section{Proof Techniques}

\subsection{Upper Bounds: Theorem \ref{thm:dp_optim_intro} (and Theorem \ref{thm:dp_rank_k})}

\noindent
Given  $M = \sum_{i=1}^n x_i x_i^\ast$ for a dataset $\{x_1,\ldots, x_n\} \subseteq \mathbb{C}^n$, where  $\|x_i\|\leq 1$ for each $i$, and a diagonal matrix $\Lambda$, the goal of our algorithm is to output a matrix $H \in \mathcal{O}_\Lambda$ which maximizes the utility $\langle M, H \rangle$ under the constraint that the output is $\eps$-differentially private.
Moreover, we would like our algorithm to run in time polynomial in the number of bits needed to represent $M$ and $\Lambda$.

\paragraph{Privacy guarantee.} 
Given data sets $\{x_i\}_{i=1}^n$ and  $\{x_i'\}_{i=1}^n$, we say that two matrices $M = \sum_{i=1}^n x_i x_i^\ast$ and $M' = \sum_{i=1}^n x_i' x_i'^\ast$ are neighbors if $x_i = x_i'$ for all but one pair of points $i$.
And we say that the output of any algorithm $\mathcal{A}$ is $\eps$-differentially private if for any $M$, $M'$ which are neighbors, and any set $S$ in the output space of the algorithm, we have
 $   \mathbb{P}(\mathcal{A}(M) \in S) \leq e^\eps  \, \,  \mathbb{P}(\mathcal{A}(M') \in S)
$.
Our algorithm ensures that its output is $\eps$-differentially private by applying the exponential mechanism of \cite{mcsherry2007mechanism} to sample a matrix $H = U \Lambda U^\ast$, where $U$ is a unitary matrix, from the unitary orbit $\mathcal{O}_\Lambda$.
For any choice of query function $q(M,H)$ and $\Delta>0$, a sample from the exponential mechanism with probability distribution proportional to $\exp\left(\frac{\eps q(D,r)}{2\Delta}\right)$, is guaranteed to be $\eps$-differentially private as long as $\Delta$ is no greater than the sensitivity 
$$ \sup_{\substack{M, M' \\ M, M' \text{ are neighbors}}} |q(M,H) - q(M',H)|$$
of the query function for all $H$.
To ensure that matrices $H$ with a larger utility $\langle M, H \rangle$ are sampled with a higher probability, we apply the exponential mechanism with the query function $q(M, H) = \langle M, H \rangle$, and sample $H$ from the distribution $\exp(\frac{\eps}{\lambda_1}\inner{\mM}{\mH}) \mathrm{d}\mu_{\Lambda}$, where $ \mathrm{d}\mu_{\Lambda}$ is a unitarily invariant measure on $\mathcal{O}_\Lambda$ obtained from the Haar measure on the unitary group.
Since we show that whenever $M$ and $M'$ differ by only one point $x_i$ 
$|\langle M,H \rangle - \langle M',H \rangle| = |x_i H x_i^\ast - x_i' H x_i'^\ast| \leq \lambda_1$ (Lemma \ref{lem:dp_optim_sensitivity}), the sensitivity  is $\Delta \leq \lambda_1$.
Thus, Algorithm \ref{alg:dp_optim_HCIZ} is $\eps$-differentially private.

\paragraph{Running time.} To generate the sample from the distribution $\nu(H) \propto \exp(\frac{\eps}{\lambda_1}\inner{\mM}{\mH}) \mathrm{d}\mu_{\Lambda}$, we use the Markov chain sampling algorithm from \cite{leake2020polynomial} (improved in \cite{MV21}), which generates a sample from the log-linear distributions on unitary orbits. The distribution $\pi$ of the output of this algorithm is guaranteed have sampling error at most $O(\eps)$ in the infinity-distance metric, $\sup_H |\log \frac{\nu(H)}{\pi(H)}| < \eps$.
Thus, the output of the Markov chain sampling algorithm is $O(\eps)$-differentially private as well.
Its running time bound is polynomial in $\lambda_1$, $\gamma_1- \gamma_d$ and the number of bits needed to represent $\vlambda=(\lambda_1, \lambda_2, \ldots, \lambda_k)$ and $\vgamma=(\gamma_1, \gamma_2, \ldots, \gamma_d)$.

\paragraph{Upper bound on error.}
Our upper bound on error is based on a covering number argument.
For any set $S$ and any $\zeta>0$, we define a $\zeta$-covering of the set $S$ with respect to a norm $\vvvert \cdot \vvvert$ on this set to be any collection of balls $\{B_1, \ldots, B_J\}$ of radius $\zeta$ with centers in $S$, where $J \in \mathbb{N}$ such that $S \subseteq \bigcup_{i=1}^J B_i$.
We define the covering number $N(S, \vvvert \cdot \vvvert, \zeta)$ to be the smallest number $J$ of Balls in any $\zeta$-covering of $S$.
The packing and covering numbers are equal up to a factor of $2$ in the radius $\zeta$ (see e.g. chapter 3.5 of \cite{mohri2018foundations}):
\begin{equation}\label{eq_cover_pack}
    P(S, \vvvert \cdot \vvvert, 2\zeta) \leq N(S, \vvvert \cdot \vvvert, \zeta) \leq P(S, \vvvert \cdot \vvvert, \zeta) \qquad \forall \zeta>0.
\end{equation}

\noindent
From a standard result about the exponential mechanism  (\cite{mcsherry2007mechanism}), we have that the utility of the exponential mechanism satisfies
\begin{equation} \label{eq_exponential_mech_utility}
    \mathbb{P}(M \notin S_t) \leq  \frac{\exp(-\frac{\eps}{2 \Delta}t)}{\mu_\Lambda(S_{\frac{t}{2}})},
\end{equation}
where $S_t$ is the set of all matrices $M$ with utility $\langle M, H \rangle > \mathrm{OPT} - t$ and $\mathrm{OPT} = \sum_{i=1}^d \lambda_i \gamma_i$ is the optimal value that $\langle M, H \rangle $ can take.
The key ingredient we need to bound the utility is an upper bound on the volume $\mu_\Lambda(S_{\frac{t}{2}})$ in the denominator of \eqref{eq_exponential_mech_utility}.
We bound this quantity via a covering number argument.
First, we show  that $S_{\frac{t}{2}}$ is contained in a spectral norm ball $B$ of radius $\frac{t}{2 \Gamma}$, where $\Gamma:= \mathrm{tr}(M)$, with center at the optimal point $H_0$, since, whenever $\|H - H_0\| \leq \frac{\tau}{2 \sum_i \gamma_i}$,
\begin{equation*}
   \langle M, H \rangle =  \langle M, H_0 \rangle -  \langle M, H_0 \rangle \geq \sum_{i=1}^d \lambda_i \gamma_i - \|H_0 - H\|_2\mathrm{tr}(M) \geq \sum_{i=1}^d \lambda_i \gamma_i  - \frac{t}{2}.
\end{equation*}
To obtain a bound on the volume of $\mu_\Lambda(B)$, we use the fact that the spectral norm $\|\cdot\|_2$ and the measure $\mu_\Lambda(B)$ are both unitarily invariant.  
We say a norm $\vvvert \cdot \vvvert$ is {\em unitarily invariant} if $\vvvert UXV \vvvert = \vvvert X \vvvert$ for any $X \in \mathbb{C}^{d\times d}$ and any unitary matrices $U,V \in \mathrm{U}(d)$; in particular $\|\cdot\|_2$ and $\|\cdot\|_F$ are unitarily invariant norms.
And we say a measure $\mu$ is unitarily invariant if $\mu(U S V) = \mu(S)$ for each subset $S$ and each $U,V \in \mathrm{U}(d)$.
Since $\mu_\Lambda$ and $\| \cdot \|_2$ are both unitarily invariant, every $\| \cdot \|_2$-norm ball of radius $\frac{t}{2 \Gamma}$ in $\mathcal{O}_\Lambda$ has the same volume with respect to the measure $\mu_\Lambda$.
Thus, if we can find a covering of $\mathcal{O}_\Lambda$ of some size $N$ consisting only of balls of radius $\frac{t}{2 \Gamma}$, we would have 
$\mu_\Lambda(B) \geq \frac{1}{N}$.
Thus, in terms of the covering number, we can rewrite the utility bound  \eqref{eq_exponential_mech_utility} as
\begin{equation} \label{eq_utility_overview_1}
 \mathbb{P}\left(\sum_{i} \gamma_i \lambda_i - \langle M , H \rangle \leq t\right) \geq N\left(\mathcal{O}_\Lambda, \|\cdot \|_2, \frac{t}{2 \Gamma} \right)  \exp\left(-\frac{\eps}{2 \Delta}t\right).
\end{equation}
To bound the utility with \eqref{eq_utility_overview_1}, we will show that the covering number of $\mathcal{O}_\Lambda$ satisfies $N\left(\mathcal{O}_\Lambda, \|\cdot \|_2, \zeta \right) \leq (1+ \frac{4 \lambda_1}{\zeta})^{2dk}$ (Lemma \ref{lem:covering_number}).
Plugging our covering number bound, and the sensitivity bound $\Delta \leq \lambda_1$ into \eqref{eq_utility_overview_1} we get that 
\begin{equation*}
   \mathbb{P}\left(\sum_{i} \gamma_i \lambda_i - \langle M , H \rangle \leq t\right) \geq  \left(1+ 8 \lambda_1 \Gamma t^{-1}\right)^{2dk} \exp\left(-\frac{\eps}{2 \lambda_1}t\right), \qquad \forall t >0.
\end{equation*}
\noindent Plugging  $t = \Theta\left(\frac{\lambda_1}{\eps} dk \log(\frac{\Gamma}{\beta})\right)$, we get that $\sum_{i} \gamma_i \lambda_i - \langle M , H \rangle \leq \tilde{O}(\frac{\lambda_1}{\eps} dk)$ w.p.  at least $1-\beta$.

In the rank-$k$ covariance matrix estimation problem, the algorithm is not handed the eigenvalues $\lambda_1,\ldots, \lambda_d$ as private information.
The Algorithm \ref{alg:dp_rank_k} in Theorem \ref{thm:dp_rank_k} perturbs the eigenvalues by adding random Laplace noise.
The proof of Theorem \ref{thm:dp_rank_k}, in addition to the proof of Theorem \ref{thm:dp_optim_intro}, requires us to carefully bound the distance between the eigenvalues $\lambda_i$ of the covariance matrix and the perturbed eigenvalues $\tilde{\lambda}_i$; see Section \ref{Appendix_covariance_estimation}.

 \paragraph{Bounding the covering number of $\mathcal{O}_\Lambda$.}
To bound the covering number of $\mathcal{O}_\Lambda$, we will first show a covering bound for the set $S_k$ of $d\times k$ matrices with orthonormal columns, and then construct a map from $S_k$ to the unitary orbit $\mathcal{O}_\Lambda$. %
Towards this end, we observe that the matrices in $\mathcal{O}_\Lambda$ are of the form $H= U \Lambda U^*$ where $U$ is a unitary matrix, and, since $\Lambda$ has only $k$ nonzero eigenvalues, $H$ only depends on the first $k$ eigenvectors of $U$, which we denote by $U_1$. 
To bound the $\zeta$-covering number of the space $S_k$ of $d \times k$ rectangular matrices $U_1$ with orthonormal columns, observe that each $U_1 \in S_k$ has spectral norm $\|U_1 \|_2$ at most 1.
Thus, the set $S_k$ of $d\times k$ complex matrices is the unit sphere in a $2dk$-dimensional (real) normed space.
To bound the covering number of $S_k$, we apply a well-known result (see e.g., Lemma 6.27 in \cite{mohri2018foundations}) which says that a minimal $\zeta'$-covering $B_1, \ldots B_t$ of the unit ball in any $2dk$-dimensional normed space has cardinality at most $(1+ \frac{2}{\zeta'})^{2dk}$.
To obtain a covering with balls with centers on the unit sphere, we take any point $x$ in $B_i \cap S_k$ (if such a point exists), and note that the ball centered at $x$ of radius $2\zeta$ contains $B_i$.

To obtain a covering of $\mathcal{O}_\Lambda$, we consider the map $\phi$ which maps each $U_1 \in S_k$ to a matrix $\phi(U_1) = U_1^\ast \Lambda U_1$.
Since $\Lambda$ has rank-$k$, $\phi: S_k \rightarrow \mathcal{O}_\Lambda$ is surjective, and thus $\phi(\hat{B}_1), \ldots, \phi(\hat{B}_t)$ is a covering of the unitary orbit; however we still need to bound the radius of the balls $\phi(\hat{B}_1)$ to show that it is a $\zeta$-covering.
Towards this end, we note that for any $U_1, U_1' \in S_k$ we have that%
\begin{equation*}
     \|\phi(U_1) -  \phi(U_1') \|_2  = \|U_1^\ast \Lambda U_1  - U_1'^\ast \Lambda U_1'\|_2 \leq  2 \|U_1^\ast \Lambda (U_1 - U_1')\|_2 \leq \frac{2}{\lambda_1} \|U_1 - U_1'\|_2.
\end{equation*}
Thus, if we set $\zeta' = \frac{\zeta}{2\lambda_1}$, we obtain a $\zeta$-covering of $\mathcal{O}_\Lambda$, and this covering has cardinality $(1+ \frac{2}{\zeta'})^{2dk} = (1+ \frac{4 \lambda_1}{\zeta})^{2dk}$, which gives an upper bound on the covering number   $N(\mathcal{O}_\Lambda, \| \cdot \|_2, \zeta)$ of $\mathcal{O}_\Lambda$.

\subsection{Lower Bounds: Theorem \ref{thm_lower_bound_general_orbit}}

To prove a lower bound on the error in the covariance matrix approximation problem, it is sufficient to consider the setting where the eigenvalues of the output matrix are given as (non-private) prior information to the algorithm.
This is because, any algorithm which works without this prior information can also be applied to this setting by simply ignoring the information about the eigenvalues of the input matrix.
Thus, any lower bound for the setting where the eigenvalues are given as a prior will also imply a lower bound for the covariance matrix estimation problem.

Towards this end, we first show a bound for a special case of the unitary orbit minimization problem (Theorem \ref{thm_lower_bound_general_orbit}), where the output matrix is in the orbit $\mathcal{O}_\Lambda$ with eigenvalues $(\lambda_1 , \ldots, \lambda_d) = \mathrm{diag}(\Lambda)$ that are equal to the (non-private) eigenvalues of the input matrix (for simplicity, in this proof overview we assume that  $\lambda_1\geq \Omega(\sqrt{d})$).
We then show that, roughly speaking, since the matrix $H$ which minimizes the Frobenius norm distance $\|M-H\|_F$ is the matrix $H = M$ and is therefore in the orbit $\mathcal{O}_\Lambda$ of the input matrix $M$, our lower bound for the unitary orbit minimization problem also implies the same lower bound for the covariance matrix estimation problem (Corollary \ref{cor_lower_covariance_estimation}; see the end of this section for an overview of the proof of this corollary).

Our lower bound relies on a ``packing number" lower bound for the orbit $\mathcal{O}_\Lambda$.
As a first attempt, we consider a maximal $\zeta$-packing of the orbit $\mathcal{O}_\Lambda$, 
$\{U_i \Lambda U_i^*\}_{i=1}^{\mathfrak{p}}$, where $\mathfrak{p} = P(\mathcal{O}_\Lambda, \|\cdot \|_F, \zeta)$ is the packing number of $\mathcal{O}_\Lambda$.
We show, using a contradiction argument, that $\|M - \mathcal{A}(M)\|_F^2 \geq \zeta^2$ (with probability at least $\frac{1}{2}$) for any input matrix $M$ and any $\zeta$ small enough such that
\begin{equation} \label{eq_lower_overview_0}
  4 \eps D^2 \leq \log P(\mathcal{O}_\Lambda, \| \cdot \|_F, \zeta),
 \end{equation}
where $D$ is the diameter of $\mathcal{O}_\Lambda$.
Suppose, on the contrary, that for every $i \in [\mathfrak{p}]$, we have that
\begin{equation}\label{eq_lower_overview_1}
    \|M_i - \mathcal{A}(M_i)\|_F^2 < \zeta^2
\end{equation}
with probability at least $\frac{1}{2}$.  
Next, observe that one can always find $m \leq \|M_i - M_j\|_F^2 +d$ vectors $x_1,\ldots, x_m$ with norm $\|x_i\| \leq 1$ such that $M_i - M_j = \sum_{s=1}^m x_s x_s^\ast + d$.
Thus, if $M_i$ and $M_j$ are data matrices with unit-norm data points, one can transform $M_i$ into $M_j$ by modifying at most $2m$ points in the dataset.
From \eqref{eq_lower_overview_1}, we have that for each $i$, the output $\mathcal{A}(M_i)$ is in the ball $B(M_i, \zeta)$ with probability at least $\frac{1}{2}$.
Thus, since $\mathcal{A}$ is $\eps$-differentially private, for every $i,j$ we have
\begin{equation}\label{eq_lower_overview_2}
  e^{-2\eps D^2} \leq e^{-\eps(\|M_i - M_j\|_F^2 +d)} \leq \frac{\mathbb{P}(\mathcal{A}(M_i) \in B(M_j, \zeta))}{\mathbb{P}(\mathcal{A}(M_i) \in B(M_i, \zeta))} \leq 2\mathbb{P}(\mathcal{A}(M_i) \in B(M_j, \zeta)) 
\end{equation}
since $M_i, M_j \in \mathcal{O}_\Lambda$, and $D$ is the diameter of  $\mathcal{O}_\Lambda$.
Thus, since \eqref{eq_lower_overview_2} holds for every $j \in \mathfrak{p}$, 
\begin{equation} \label{eq_lower_overview_3}
  1 \geq \sum_{j=1}^{\mathfrak{p}}\mathbb{P}(\mathcal{A}(M_i) \in B(M_j, \zeta))  \geq 2\mathfrak{p}\times e^{-2 \eps D^2}.
\end{equation}
Rearranging \eqref{eq_lower_overview_3}, we get that $\eps D^2 \leq \log \mathfrak{p}= \log P(\mathcal{O}_\Lambda, \| \cdot \|_F, \zeta)$, which contradicts our assumption in \eqref{eq_lower_overview_0}.
Thus, by contradiction, we have that $\|M - \mathcal{A}(M)\|_F^2 \geq \zeta^2$ with probability $\Omega(1)$ for any $\zeta>0$ satisfying \eqref{eq_lower_overview_0}.

Plugging in our packing number bound for $\mathcal{O}_\Lambda$ (\eqref{eqref_packing_bound_entire_orbit} in Theorem \ref{lemma_orbit_packing_improved}), and solving for the largest value of $\zeta$ satisfying \eqref{eq_lower_overview_0}, gives the lower bound of
 \begin{equation} \label{eq_lower_overview_4}
 \|M - \mathcal{A}(M)\|_F^2 \geq c\hat{D}^2
    \times(\lambda_i - \lambda_{j})^2 \exp\left(\frac{ -2\eps D^2}{i\times (d-j+1)} \right)
    \end{equation}
    for every $1\leq i,j \leq d$,  where $\hat{D}:=D_{\| \cdot\|_F }(\mathcal{G}_{d-j+i+1,\, \,i})$ is the diameter of the Grassmannian.
    Unfortunately, since the diameter of $\mathcal{O}_\Lambda$ is $D \geq \lambda_1 - \lambda_d$ this lower bound is {\em exponential} in $\lambda_1 - \lambda_d$.
    The $D$ term in the exponent comes from the fact that, since we have used a packing for the entire unitary orbit $\mathcal{O}_\Lambda$, the distance between any two balls in our packing is upper bounded by $D$.
    To achieve a bound that is polynomial in $\lambda_1 - \lambda_d$, we would instead like to use a packing for a smaller {\em subset} of the orbit $\mathcal{O}_\Lambda$, of some radius (roughly) $\omega \leq \sqrt{\frac{i\times (d-j+1)}{ -2\eps}}$.
    However, restricting our packing to a ball of radius $\omega$-- rather than the entire orbit--  requires us to prove a packing number for a subset of the orbit, (\eqref{eqref_packing_bound_subset}).
    This leads to additional challenges in the proof of our packing bound which we describe in the next subsection.
    
    Replacing the $D$ term in \eqref{eqref_packing_bound_entire_orbit} with $\omega = \Theta( \frac{\lambda_1}{\lambda_i-\lambda_j} \zeta)$ and plugging in our bound for the $\zeta$-packing number of a ball of radius $\omega$ inside the orbit, and solving for the largest value of $\zeta$ satisfying \eqref{eq_lower_overview_0}, gives the improved lower bound of
    \begin{equation}\label{eq_lower_overview_5}
         \|M - \mathcal{A}(M)\|_F^2 \geq \Omega\left(\frac{i\times (d-j+1)}{\lambda_1^2 \eps} \times (\lambda_i-\lambda_j)^2\right)
             \end{equation}
             for every $1 \leq i < j \leq d$.
             Unlike the bound in \eqref{eq_lower_overview_4} which is exponential in the $\lambda's$, this bound is {\em polynomial} in the $\lambda$'s and in $d, \frac{1}{\eps}$.
             If we plug in $j = d$ in \eqref{eq_lower_overview_5} and take the maximum over all $i \in [d]$, and then plug in $i = 1$ and take the maximum over all $j \in [d]$, and finally take the larger of these two maximum values, we recover the error lower bound of Theorem \ref{thm_lower_bound_general_orbit}.

\subsection{Packing Number Lower Bounds: Theorem \ref{lemma_orbit_packing_improved}}
In this section we first explain how we bound the packing number of the entire unitary orbit $\mathcal{O}_\Lambda$ (\eqref{eqref_packing_bound_entire_orbit}), and we then explain how we extend the proof to obtain a bound on the packing number of any ball inside $\mathcal{O}_\Lambda$ (\eqref{eqref_packing_bound_subset}).

The general strategy for proving our packing bounds for the unitary orbit (Theorem \ref{lemma_orbit_packing_improved}), is to first construct a map $\phi: \Omega \rightarrow \mathcal{O}_\Lambda$ from some space $\Omega$ with previously known packing number bounds to the unitary orbit. 
And, once we have a map $\phi$ and a packing $X_1, \ldots, X_t \in \Omega$, we show that the map preserves (a lower bound for) distances between points in the packing: $\|\phi(X_i) - \phi(X_j)\|_F \geq \beta \|X_i - X_j\|_F$ for some $\beta>0$, implying  $\phi(X_1), \ldots, \phi(X_t)$ is a $\zeta \beta$-packing of $\mathcal{O}_\Delta$. 

As a first attempt, we consider the space of unitary matrices $\mathrm{U}(d)$ for our choice of $\Omega$, and the map $\phi: U \rightarrow U \Lambda U^\ast$.
Unfortunately, there may be $U, U' \in \mathrm{U}(d)$ such that $\|\phi(U) - \phi(U')\|_F = 0$ even though $\|U - U'\|_F>0$ (For instance, if $\mathrm{diag}(\Lambda) = (1,1,0)$ and $U = I$ and $U'$ is the matrix $[e_2, e_1, e_3]^\top$ where $e_i$ is the vector with a $1$ in the $i$'th entry and zero everywhere else, we have $\phi(U) - \phi(U') =0$ and yet $\|U - U'\|_F = 2$.).

To get around this problem we instead consider a map $\phi$ from the (complex) Grassmannian manifold $\mathcal{G}_{d,i}$, the collection of subspaces of dimension $i$ in $d$-dimensional space, to $\mathcal{O}_\Lambda$.
Identifying each subspace with its associated rank-$i$ projection matrix, we construct a maximal $\zeta$-packing for $P_1, \ldots, P_{\mathfrak{p}} \in \mathcal{G}_{d,i}$, where $\mathfrak{p}$ is the packing number of $\mathcal{G}_{d,i}$.
To bound the size of this packing, we use the covering/packing number bound from \cite{szarek1982nets} for the  Grassmannian $\mathcal{G}_{d,i}$, which says that $$
\mathfrak{p} =: P(\mathcal{G}_{d,i}, \| \cdot \|_F, \zeta) \geq \left(\zeta^{-1}c D_{\|\cdot\|_F}(\mathcal{G}_{d,i})) \right)^{2d i},$$ where $D_{\|\cdot\|_F}(\mathcal{G}_{d,i}))$ is the diameter of $\mathcal{G}_{d,i}$ and $c$ is a universal constant.

To define our map $\phi(P)$ for any rank-$i$ projection matrix $P \in \mathcal{G}_{d,i}$, we find  a $d \times i$ matrix $U_1$ whose columns form an orthonormal basis for the space spanned by the columns of $P$; thus, $U_1 U_1^\ast = P$
(for now, we choose the matrix $U_1$ in an arbitrary manner, although we will choose $U_1$ more carefully for our proof of \eqref{eqref_packing_bound_subset}).
We also find a $d \times (d-i)$  matrix $U_2$ whose columns are orthogonal to the columns of $U_1$.
Thus, $[U_1, U_2]$ is a unitary matrix.
This allows us to define the map $\phi$ by $\phi(P) = U \Lambda U^\ast$, where $U= [U_1, U_2]$.

To show that $\phi$ preserves a lower bound on the Frobenius norm distance, use the $\sin$-$\Theta$ theorem of \cite{davis1970rotation} (Lemma  \ref{lemma_SinTheta}) which gives a bound on how much the eigenvectors of a Hermitian matrix can ``rotate" when the matrix is perturbed.
More specifically, the $\sin$-$\Theta$ theorem says that if $A, A'$ are Hermitian matrices, with eigenvalues $\lambda_1,\ldots, \lambda_d$ and $\lambda_1', \ldots, \lambda_d'$, and $V_1$ and $V_1'$ are the matrix whose columns are the first $i$ eigenvectors of $A$ and $A'$ respectively, then $\|V_1 V_1^\ast - V_1' V_1'^\ast\|_F \leq \frac{\|A-A'\|_F}{\lambda_i - \lambda_{i+1}'}$.
Applying the $\sin$-$\Theta$ Theorem, for any $P, P' \in \mathcal{G}_{d,i}$ we have
\begin{equation} \label{eq_packing_overview_1}
    \|\phi(P)- \phi(P')\|_F =\|U \Lambda U^\ast - U' \Lambda U'^\ast\|_F \geq  
    (\lambda_i - \lambda_{i+1}) \times \|P- P'\|_F,
\end{equation}
for some unitary matrices $U=[U_1, U_2]$ and $U'=[U_1', U_2']$ such that $P= U_1 U_1^\ast$ and $P' = U_1' U_1'^\ast$.

Inequality \ref{eq_packing_overview_1} implies that since $P_1, \ldots P_{\mathfrak{p}}$ is a $\zeta$-packing of $\mathcal{G}_{d,i}$,  $\phi(P_1), \ldots \phi(P_{\mathfrak{p}})$ must be a $\zeta \times (\lambda_i - \lambda_{i+1})$-packing of $\mathcal{O}_\Lambda$.
Thus \eqref{eq_packing_overview_1}, together with the bound on the packing number of $\mathcal{G}_{d,i}$, gives the following bound on the packing number of $\mathcal{O}_\Lambda$
\begin{equation}\label{eq_packing_overview_2}
  P(\mathcal{O}_\Lambda, \| \cdot \|_F, \zeta) \geq \left( \zeta^{-1}cD_{\|\cdot\|_F}(\mathcal{G}_{d,i}) \times  (\lambda_i - \lambda_{i+1}) \right)^{2d i} \qquad \forall i \in [d].
\end{equation}

\paragraph{Improving the packing lower bound.}
While \eqref{eq_packing_overview_2} gives a bound for the $\zeta$-packing number of $\mathcal{O}_\Lambda$,
the eigenvalue gap term $\lambda_i - \lambda_{i+1}$ may be much smaller than the eigenvalue gap term $\lambda_i - \lambda_j$ which appears in the packing number bounds we ultimately show in Theorem \ref{thm_lower_bound_general_orbit}.

To get around this problem, we replace the map $\phi: \mathcal{G}_{d,i} \rightarrow \mathcal{O}_\Lambda$ and instead consider a more general map $\phi: \mathcal{G}_{d-j+i+1,i} \rightarrow \mathcal{O}_\Lambda$ for any $i,j \in [d]$.
Namely, for any $(d-j+i+1)\times (d-j+i+1)$ rank-$i$ projection matrix $P \in \mathcal{G}_{d-j+i+1,i}$, we choose a matrix $U_1$ with orthonormal columns such that $U_1 U_1^\star = P$, and choose $U_2$ such that $[U_1, U_2]$ is a $(d-j+i+1)\times (d-j+i+1)$  unitary matrix.
And, denoting by A[i:j] the rows $i,\ldots, j$ of a given matrix $A$, we set
\begin{equation*}
   U = \left({\begin{array}{ccc}
   U_1[1:i] & 0 &  U_2[1:i]\\
    0 & I & 0 \\
     U_1[i+1:d-j+1]  & 0 & U_2[i+1:d-j+1] \\
  \end{array}}  \right) \in \mathrm{U}(d),
\end{equation*}
and set $\phi(P) = U \Lambda U^\ast$.
Then, denoting $\tilde{\Lambda} = \mathrm{diag}(\lambda_1,\ldots, \lambda_i, \lambda_j, \ldots, \lambda_d)$,we have  by the $\sin$-$\Theta$ theorem, for any $P,P' \in \mathcal{G}_{d-j+i+1,i}$, that
\begin{equation} \label{eq_packing_overview_3}
    \|\phi(P) - \phi(P')\|_F = \|\hat{U} \tilde{\Lambda} \hat{U}^\ast - \hat{U}' \tilde{\Lambda} \hat{U}'^\ast\|_F \geq (\lambda_i - \lambda_j)\|P - P'\|_F ,
\end{equation}
 for unitary matrices $\hat{U} =[U_1, U_2]$ and $\hat{U}' =[U_1', U_2']$ such that $P=\hat{U} \hat{U}^\ast$ and $P'=\hat{U}' \hat{U}'^\ast$.
Combining \eqref{eq_packing_overview_3} with the lower bound for the packing number of the Grassmannian $\mathcal{G}_{d-j+i+1,i}$, we obtain our bound on the packing number for the unitary orbit $\mathcal{O}_\Lambda$ (Inequality \ref{eqref_packing_bound_entire_orbit} in Theorem \ref{thm_lower_bound_general_orbit}). 

\paragraph{Packing number lower bounds for $B\cap \mathcal{O}_\Lambda$.} 
To obtain a packing number bound for a subset of the unitary orbit $B(X, \omega)\cap \mathcal{O}_\Lambda$ where $B(X, \omega)$ is some ball of radius $\omega$ with center $X\in \mathcal{O}_\Lambda$, we need to ensure that our packing lies inside a ball of radius $\omega$.
Towards this end, we first extend the packing number lower bound of \cite{szarek1982nets} for the Grassmannian $\mathcal{G}_{d,i}$, to a packing number  lower bound for a ball inside the Grassmannian via a simple covering argument (Lemma \ref{lemma_packing_subset_ball}).
Since $\|\cdot \|_F$ is unitarily invariant, the packing number is the same regardless of the center of the ball; thus, for simplicity we set the center of the ball in $\mathcal{G}_{d,i}$ to be the rank-$i$ projection matrix $I_i$ consisting of the first $i$ columns of the identity matrix. 
While we have already shown that the map $\phi$ preserves a lower bound on the Frobenius distance $\|\phi(P) - \phi(P')\|_F$ between points in the packing (Inequality \ref{eq_packing_overview_1}) to obtain a packing inside $\mathcal{O}_\Lambda$, in order to ensure that the packing lies inside a ball of radius $\omega$ we will also need to show that the map $\phi$ preserves an {\em upper} bound on this distance.

Unfortunately, if we construct the map $\phi(P)$ by choosing the columns of $U_1$ to be an arbitrary orthonormal basis for the column space of $P$ and then set $\phi(P) = U \Lambda U^\ast$ where $U$ is an arbitrary unitary matrix whose first $i$ columns are $U_1$,  we may have that $\|\phi(P) - \phi(P')\|_F>1$ even when the distance $\|P-P'\|_F$ is arbitrarily small (e.g., if $\mathrm{diag}(\Lambda) = (2,1,0)$, $U=I$, and $U' =[e_2, e_1, e_3]R_\theta$, where $R_\eta$ is a rotation matrix for a small angle $\eta>0$, and we choose $U_1$ to be the first 2 columns of $U$ and $U'$ respectively,  we have $\phi(U_1 U_1^\ast) - \phi(U_1' U_1'^\ast) >1$ and yet $\|U_1 U_1^\ast - U_1' U_1'^\ast\|_F = \eta$.).
This is because there are many ways to choose the basis $U_1$ for the column space of $P$.

To show a lower bound on $\|\phi(P) - \phi(P')\|_F$, when constructing the map $\phi(P)$ we will choose the eigenvectors $U_1$ of $P$ such that, roughly speaking, they correspond to the ``principal vectors" between the subspaces spanned by the columns of  $P$ and the columns of the projection matrix $P_0 = I_i$ which $\phi$ maps to the center $X=\Lambda$ of the ball $B$.
We define the principle vectors and principle angles $\theta_1, \ldots, \theta_{i}$ between any two $i$-dimensional subspaces $\mathcal{U}$ and $\mathcal{V}$ recursively starting with $\ell=1$ as follows (see e.g.  \cite{bjorck1973numerical}):
\begin{equation} \label{eq_principle_angles}
   \theta_\ell = \min\left\{\frac{\arccos|\langle u, v \rangle|}{\|u\| \|v\|} : u \in \mathcal{U}, v \in \mathcal{V}, u \perp u_s, v \perp v_s \forall s \in{1, \ldots, \ell-1}\right\}.
\end{equation}
Letting $\mathcal{U}$ be the subspace spanned by the columns of any rank-$i$ projection matrix $P$, and $\mathcal{V}$ the subspace spanned by the columns of $P_0 = I_i$, we set $V_1 =[v_1,\ldots, v_i]$ and $U_1 = [u_1,\ldots, u_i]$ to be the principle vectors between the two subspaces.
Thus, roughly speaking, \eqref{eq_principle_angles} implies that we have chosen matrices $U_1$ and $V_1$ with the smallest possible angles between the columns of $U_1$ and the corresponding columns of $V_1$ under the constraint that $U_1 U_1^\ast = I_i$ and $V_1 V_1^\ast = P$.
We then define the map to be  $\phi(P) = W^\ast \Lambda W^\ast$ where $W$ is a unitary matrix whose first $i$ columns are $V_1 U_1^\ast$, and the last $d-i$ columns are obtained using a similar ``principle angle" construction as the first $i$ columns.
In particular, we have $\phi(P_0) = \Lambda$.

We then show
$
\|U_1-V_1\|_F^2 =  2k -2\sum_{\ell=1}^i \cos(\theta_\ell) \leq \|V_1V_1^\ast - I_i \|_F^2,
$
and hence (Lemma \ref{lemma_unitary_to_projection}),
\begin{equation*}
 \|V_1 U_1^\ast \hat{I}_i - \hat{I}_i\|_F  \leq     \|V_1V_1^\ast - I_i \|_F = \|P - I_i \|_F
\end{equation*}
where $\hat{I}_i$ is the first $i$ columns of the identity matrix.
This in turn implies the bound
\begin{equation*}
  \|\phi(P) - \Lambda\|_F \leq 2 \lambda_1 \| W - I\|_F \leq 4 \lambda_1 \|P - I_i \|_F.
\end{equation*}
We now have an upper bound on the distance $\|\phi(P) - \Lambda\|_F$ between any matrix $\phi(P)$ in our packing and the center $\Lambda$ of the ball $B(\Lambda, \omega)$ we would like to pack.
Combining this bound with our lower bound on $\|\phi(P) - \phi(P')\|_F$ of the previous subsection allows us to show our packing number lower bound for the ball  $B(\Lambda, \omega) \cap \mathcal{O}_\Lambda$ (Inequality \ref{eqref_packing_bound_subset} of Theorem \ref{lemma_orbit_packing_improved}).

\section{Preliminaries}
\label{sec:prelim}

\subsection{Notation}
For any vector $\vv \in \C^d$, we denote by $\|\vv\|$  its Euclidean ($\ell_2$-norm) and by $\|\vv\|_p$  its $\ell_p$-norm.
For any matrix $\mM \in \C^{m\times n}$, we denote by $\|\mM\|$  its spectral norm ($\ell_2$-operator norm), by $\|\mM\|_p$  its $\ell_p$-operator norm, and by $\|\mM\|_F$  its Frobenius norm.
We use the standard definition in the Euclidean space for inner products.
For two vectors $\vu, \vv \in \C^d$, we denote the inner product of them as $\inner{\vu}{\vv} \coloneqq \vu^*\vv$.
For two matrices $\mM, \mN \in \C^{m\times n}$, we denote their Frobenius inner product by $\inner{\mM}{\mN}\coloneqq \Tr\left(\mM^*\mN\right)$.
For any $d \in \sZ_+$, we denote by $\gS_+^d \subset \R^d$ the set of $d\times d$ positive semi-definite (PSD) real matrices.
For any $d \in \sZ_+$, we denote by $\gH_+^d \subset \C^d$  the set of $d\times d$ PSD Hermitian matrices.

\subsection{Preliminaries on Differential Privacy}
The Laplace distribution  with mean $0$ and parameter $b$ is defined over $\mathbb{R^d}$ as
$\mathrm{Lap}(x) \coloneqq \frac{1}{2b}e^{-|x|}$.

\begin{defn}[\textbf{Sensitivity}]
\label{def:sensitivity}
Given collection of datasets $\gD$ with a notion of neighboring datasets, the sensitivity of a query function $q: \gD \to \R^d$ is denoted by $\Delta q$ and defined as
\begin{equation*}
   \Delta q \coloneqq \sup_{\substack{D, D' \in \gD\\ D, D' \text{ are neighbors}}}\|q(D)-q(D')\|_1.
\end{equation*}
\end{defn}

\begin{thm}[\textbf{Laplace mechanism  and its differential privacy \cite{dwork2006differential}}]
\label{thm:laplace}
For a given collection of datasets $\gD$ and a privacy budget $\eps > 0$, given any function $f: \gD\to \R^d$, define the Laplace mechanism $\gM:\gD\to \R^d$  as $
    \gM(D) \coloneqq f(D) + (Y_1, \ldots, Y_d),
$
where  $Y_i$'s are i.i.d. random variables drawn from $\Lap(\Delta f/\eps)$.
Then, $\mathcal{M}$  is $\eps$-differentially private.
\end{thm}

\begin{thm}[\textbf{Exponential mechanism  \cite{mcsherry2007mechanism}}]
\label{thm:exponential}
For a given collection of datasets $\gD$ with a notion of neighboring datasets, a measurable set of all possible results $\gR$, and a privacy budget $\eps > 0$, given any query function $q: \gD \times \gR \to \R$, define the exponential mechanism $\gM:\gD \to \gR$ as follows: For any dataset $D\in \gD$, $\gM(D)$ outputs an $r\in \gR$ sampled from a distribution with probability density proportional to $$\exp\left(\frac{\eps q(D,r)}{2\Delta q}\right).$$
Then,
$\gM$  is $\eps$-differentially private.
\end{thm}

\begin{thm}[\textbf{Utility guarantee for exponential mechanism  \cite{mcsherry2007mechanism}}]
\label{thm:exponential_mechanism_utility}
As in the setting in Theorem \ref{thm:exponential}, given a dataset $D$, a query function $q$ and privacy budget $\eps$, let 
 $   S_t \coloneqq \leftbrace r: q(D,r) > {\OPT} -t \rightbrace
$,
where $\OPT \coloneqq \max_r q(D,r)$.
Then, we have 
$$    \PR[r \not\in S_t] \leq \frac{\exp\left(-\frac{\eps}{2\Delta q} t\right)}{\mu(S_{t/2})},$$
where $\mu$ is the base measure of the $\gR$, the set of all possible results.
\end{thm}

\noindent
\paragraph{Neighboring datasets.} In our setting,  $\mathcal{U}$ is the universe of users. For each $u \in \mathcal{U}$, we have 
a vector $v_u \in \mathbb{C}^d$ such that $\|v_u\|_2 \leq 1$.
Given a dataset $D \subseteq \mathcal{U}$, define $A\coloneqq \sum_{u \in D} v_u v_u^*$.
Two $d \times d$ Hermitian PSD  matrices $\mA$ and $\mA'$ are said to be {\em neighbors} if and only if there exists $\vu, \vv \in \C^d$ such that $\|\vu\|, \|\vv\|\leq 1$ and $\mA' = \mA-\vu\vu^*+\vv\vv^*$.

\section{Differentially Private Optimization on Orbits: Proof of Theorem \ref{thm:dp_optim_intro}}\label{Appendix_upper_utility_bound}
The proof of Theorem \ref{thm:dp_optim_intro} consists of four parts: the algorithm, its privacy guarantee, its utility guarantee, and its running time.
\subsection{Algorithm}
We first present the algorithm in Theorem \ref{thm:dp_optim_intro}.
\begin{algorithm}[ht]
\caption{Differentially private unitary orbit approximation}
\label{alg:dp_optim_HCIZ}
\SetKwInOut{KwInput}{Input}
\SetKwInOut{KwOutput}{Output}
\SetKwInOut{KwAlgorithm}{Algorithm}
\KwInput{A matrix $\mM \in \gH_+^d \subset \C^{d\times d}$ with eigenvalues $\gamma_1 \geq \cdots \geq \gamma_d \geq 0$, the output matrix's maximum rank $k\in [d]$, a list of top $k$ eigenvalues of the output matrix $\lambda_1 \geq  \cdots \geq \lambda_k \geq 0$, a privacy budget $\epsilon > 0$}
\KwOutput{A matrix $\mH \in\gH_+^d \subset \C^{d\times d}$}
\KwAlgorithm{}
\begin{enumerate}
\item Define $\mLambda \gets \diag(\lambda_1, \ldots, \lambda_k, 0, \ldots, 0) \in \C^{d\times d}$\\
\item Sample $\mH \in \gO_{\mLambda}$ from a distribution that is $\frac{\epsilon}{4}$-close in infinity divergence distance to the \\ distribution $d\nu(\mH) \propto \exp\left(\frac{\eps}{4\lambda_1}\inner{\mM}{\mH}\right)d\mu_\Lambda(\mH)$ 
\item Output $\mH$
\end{enumerate}
\end{algorithm}
\subsection{Privacy guarantee}
To prove the privacy guarantee we first need to bound the sensitivity of the utility function $\inner{\mM}{\mH}$.
\begin{lem}[\textbf{Sensitivity bound}]
\label{lem:dp_optim_sensitivity}
Given  $d$ and a list of eigenvalues $\lambda_1 \geq \cdots \geq \lambda_k \geq 0$ for some $k\in[d]$, let $\mLambda \coloneqq \diag(\lambda_1, \ldots,\lambda_k,0,\ldots ,0) \in \C^{d\times d}$.
For any two neighboring $d\times d$ PSD Hermitian $\mA, \mA' \in \gH_+^d$ such that $\mA' = \mA - \vu\vu^* + \vv\vv^*$ for some $\vu, \vv$ such that $\|\vu\|_2,\|\vv\|_2 \leq 1$, and for any PSD Hermitian matrix $\mH \in \gO_\mLambda$, we have
\begin{equation*}
    |\inner{\mA}{\mH}-\inner{\mA'}{\mH}| \leq \lambda_1.
\end{equation*}
\label{lem:sensitivity}
\end{lem}
\begin{proof}
Since $\min(0,\lambda_k) \leq \vv^*\mH\vv \leq \max(0,\lambda_1)$ for any $\vv$ with $\|\vv\|_2 \leq 1$, 
\begin{equation*}
    |\inner{\mA}{\mH}-\inner{\mA'}{\mH}|   = |\vu^*\mH\vu- \vv^*\mH\vv| \leq   \lambda_1.
\end{equation*}
\end{proof}
With Lemma \ref{lem:sensitivity}, we can prove the privacy guarantee for Algorithm \ref{alg:dp_optim_HCIZ}.
\begin{lem}[\textbf{Privacy guarantee for Algorithm \ref{alg:dp_optim_HCIZ}}]
\label{lem:dp_optim_privacy_guarantee}
The randomized algorithm $\gM$ as described in Algorithm \ref{alg:dp_optim_HCIZ} is $\eps$-differentially private, for the given privacy budget $\eps > 0$.
\end{lem}
\begin{proof}
Given neighboring $d\times d$ PSD Hermitian matrices $\mA, \mA' \in \gH_+^d$ such that $\mA' = \mA - \vu\vu^* + \vv\vv^*$ for some $\vu, \vv$ with $\|\vu\|_2,\|\vv\|_2 \leq 1$, and any matrix $\mH \in \gO_\mLambda$, we want to bound the ratio of probability density of $\gM$ at $\mH$ for $\mA$ and $\mA'$. Let $\Tilde{\nu}_\mA(\mH)$ be the output density of $\mH$ for $\gM(\mA)$ and $\nu_\mA(\mH)$ be the target density of $\mH$, which is given by $\exp(\frac{\epsilon}{4\lambda_1}\inner{\mA}{\mH})$. We have $D_\infty(\Tilde{\nu}_\mA \| \nu_\mA) \leq \frac{\epsilon}{4}$.
\[
\begin{split}
    \frac{\nu_A(H)}{\nu_{A'}(H)} = \frac{\frac{ e^{\frac{\epsilon}{4}  \langle A, H \rangle}}{\int_{Q \in \gO_\mLambda} e^{\frac{\epsilon}{4}  \langle A, Q\rangle}d\mu_\Lambda(Q)}}{\frac{ e^{\frac{\epsilon}{4}  \langle A', H \rangle}}{\int_{Q \in \gO_\mLambda} e^{\frac{\epsilon}{4}  \langle A', Q\rangle}d\mu_\Lambda(Q)}} &= e^{\frac{\epsilon}{4}  \langle A-A', H \rangle} \cdot \frac{\int_{Q \in \gO_\mLambda} e^{\frac{\epsilon}{4}  \langle A', Q\rangle}d\mu_\Lambda(Q)}{\int_{Q \in \gO_\mLambda} e^{\frac{\epsilon}{4}  \langle A, Q\rangle}d\mu_\Lambda(Q)} \\
        &\leq e^{\frac{\epsilon}{4}} \cdot \frac{\int_{Q \in \gO_\mLambda} e^{\frac{\epsilon}{4}  \langle A, Q\rangle + \frac{\epsilon}{2} \langle A'-A, Q \rangle}d\mu_\Lambda(Q)}{\int_{Q \in \gO_\mLambda} e^{\frac{\epsilon}{4}  \langle A, Q\rangle}d\mu_\Lambda(Q)} \\
        &\leq e^{\frac{\epsilon}{4}} \cdot \max_{Q \in \gO_\mLambda} e^{\frac{\epsilon}{4} |\langle A'-A, Q \rangle|} \\
        &\leq e^{\frac{\epsilon}{2}}.
\end{split}
\]
Using the infinity divergence bounds between $\tilde{\nu}_A$ and $\nu_A$, we then further have that
\[
    \frac{\tilde{\nu}_A(H)}{\tilde{\nu}_{A'}(H)} = \frac{\tilde{\nu}_A(H) / \mu_A(H)}{\tilde{\nu}_{A'}(H) / \nu_{A'}(H)} \cdot \frac{\mu_A(H)}{\mu_{A'}(H)} \leq \frac{e^{\frac{\epsilon}{4}}}{e^{-\frac{\epsilon}{4}}} \cdot e^{\frac{\epsilon}{2}} = e^\epsilon.
\]

\begin{equation*}
    \frac{\Tilde{\nu}_\mA(\mH)}{\Tilde{\nu}_{\mA'}(\mH)} \leq \exp\left({\epsilon}\right).
\end{equation*}
\end{proof}

\subsection{Utility guarantee}
In this section we prove a guarantee on the utility of Algorithm \ref{alg:dp_optim_HCIZ}. 
Towards this end, we first prove a covering number lemma for the unitary orbit.

\begin{lem}[\textbf{Covering number for $\gO_\mLambda$}]
\label{lem:covering_number}
For any $\zeta > 0$, the covering number of $\gO_\mLambda$ is at most $N(\mathcal{O}_\Lambda, \|\cdot\|_2, \zeta) \leq  (1+\frac{8\lambda_1}{\zeta})^{2dk}$, with $\mLambda$ defined in Algorithm \ref{alg:dp_optim_HCIZ}.
\end{lem}
\begin{proof}

First consider $S_k$, the set of $k \times d$ complex matrices with orthonormal rows.
    Fix any $M \in S_k$ and let $U$ be a unitary matrix such that the first $k$ rows of $U$ are the rows of $M$.
    Letting $\|\cdot\|_2$ denote the $2 \to 2$ operator norm, we have that $\|M\|_2 = \|M^*\|_2 = 1$ since $M^*M$ is a PSD projection.
    Hence, the set $S_k$ can be considered a subset of the unit sphere in a $dk$-dimensional complex normed vector space.
    By a standard result (see e.g., Lemma 6.27 in \cite{mohri2018foundations}), we can cover the complex unit ball in such a space with respect to any norm by at most $(1 + \frac{2}{\zeta})^{2dk}$ balls of radius $\zeta$ for any $\zeta > 0$.
    By replacing each such ball $B$ with a ball of radius $2\zeta$ centered about any $M \in B \cap S_k$ (if such a point exists), we have that we can cover $S_k$ by at most $(1 + \frac{4}{\zeta})^{2dk}$ balls centered in $S_k$ of radius $\zeta$ for any $\zeta > 0$.

Consider the map $\phi: \mM \to \mM^*\diag(\vlambda)\mM$, which maps $S_k$ to $\gO_\mLambda$. Given any $\mM, \mM'$ with $\|\mM-\mM'\|_2< \frac{\zeta}{2\lambda_1}$, we have
\begin{eqnarray*}
    \|\phi(\mM)-\phi(\mM')\|_2 &=& \|\mM^*\diag(\vlambda)\mM - \mM'^*\diag(\vlambda)\mM'\|_2\\ &\leq & \|\mM^*\diag(\vlambda)(\mM-\mM')\|_2 + \|(\mM-\mM')^*\diag(\vlambda)\mM'\|_2\\
    &\leq &\lambda_1\frac{\zeta}{2\lambda_1} + \lambda_1\frac{\zeta}{2\lambda_1} \\
    &=& \zeta.
\end{eqnarray*}
Thus, for any ball $B$ with radius $\frac{\zeta}{2\lambda_1}$ centered at some $\mM\in S_k$, we have $\phi(B\cap S_k)$ contained in an $\zeta$-ball centered at $\phi(\mM)\in\gO_\mLambda$. Since $\phi$ is surjective, $\gO_\mLambda$ can be covered with at most $(1+\frac{8\lambda_1}{\zeta})^{2dk}$ balls centered in $\gO_\mLambda$ of radius $\zeta$ for any $\zeta > 0$.
\end{proof}

.

\noindent
To prove our utility bound, we need the utility bound on the exponential mechanism (Theorem \ref{thm:exponential_mechanism_utility}).
We use the notation $\Gamma \coloneqq \sum_{i=1}^d\gamma_i$. 
The following lemma assumes that we can sample exactly from the distribution proportional to $\exp(\frac{\eps}{4\lambda_1}\inner{\mM}{\mH})$.
\begin{lem}[\textbf{Probability bound assuming exact sampling for the sampling step in Algorithm \ref{alg:dp_optim_HCIZ}}]
\label{lem:probability_HCIZ_sampling}
Let the input and output be as listed in Algorithm \ref{alg:dp_optim_HCIZ}.
Assume $\mH \in \gO_\mLambda$ is sampled exactly from the distribution $\exp(\frac{\eps}{4\lambda_1}\inner{\mM}{\mH})$, then we have 
\begin{equation}
\label{eqn:probability_HCIZ_sampling}
    \PR\left[\inner{\mM}{\mH} \leq \sum_{i=1}^k\lambda_i\gamma_i - \tau\right] \leq N\left(\mathcal{O}_\Lambda, \, \, \|\cdot\|_2, \,\, \frac{\tau}{2\sum_{i=1}^d\gamma_i}\right) \times \exp\left(-\frac{\eps\tau}{4\lambda_1}\right). 
\end{equation}
\end{lem}
\begin{proof}
Using Theorem \ref{thm:exponential_mechanism_utility} in this case, we have \begin{equation*}
    S_t = \leftbrace \mH: \inner{\mM}{\mH}  >  \sum_{i=1}^k\lambda_i\gamma_i -  t \rightbrace.
\end{equation*}
Let $\mH_0 = \mU\Lambda \mU^*$, where $\mU$ is the unitary matrix obtained by diagonalizing $\mA=\mU\diag(\vgamma)\mU^*$. $\mH_0$ is the optimal output and we have $\inner{\mA}{\mH_0}  =\sum_{i=1}^k\lambda_i\gamma_i$.
We fix any $\mH \in \gO_\mLambda$ such that $\|\mH_0-\mH\|_2 \leq \frac{\tau}{2\Gamma}$. We can then apply the H\"{o}lder’s inequality to get
\begin{align*}
    \inner{\mM}{\mH} 
    &= \inner{\mM}{\mH_0} - \inner{\mM}{\mH_0-\mH}\\
    &\geq \sum_{i=1}^k\lambda_i\gamma_i - \|\mH_0-\mH\|_2\Tr(\mM)
    \tag{Using H\"{o}lder’s inequality}\\
    &= \sum_{i=1}^k\lambda_i\gamma_i -\|\mH_0-\mH\|_2\Gamma \tag{Substitute definitions}\\
    &\geq \sum_{i=1}^k\lambda_i\gamma_i - \frac{\tau}{2\Gamma}\Gamma
    \tag{Substitute $\|H_0-H\|_2$}\\
    &\geq \sum_{i=1}^k\lambda_i\gamma_i - \frac{\tau}{2}.
\end{align*}
Thus, if $\|\mH_0-\mH\|_2 \leq \frac{\tau}{2\Gamma}$, then $\mH \in   S_{\frac{\tau}{2}}$. Thus, every $\mH$ contained in the ball of radius $\frac{\tau}{2\Gamma}$ centered at $\mM$ is also in $S_{\frac{\tau}{2}}$.
Let $\mu$ be the unitarily invariant probability measure on $\gO_\mLambda$. 
By the definition of covering number, the number of balls centered in $\gO_\mLambda$ of radius $\frac{\tau}{2\Gamma}$ required to cover the set $\gO_\mLambda$ is at most $N(\mathcal{O}_\Lambda, \|\cdot\|_2, \frac{\tau}{2\Gamma})$.
Thus, there exists some ball $B_{\frac{\tau}{2\Gamma}}(\mH')$ centered at $\mH'$ with
$\mu(B_{\frac{\tau}{2\Gamma}}(\mH')) \geq N(\mathcal{O}_\Lambda, \|\cdot\|_2, \frac{\tau}{2\Gamma})^{-1}$.
Thus, since $\mu$ is unitarily invariant, 
\begin{eqnarray*}
    \mu\left(S_{\frac{\tau}{2}}\right) 
    &\geq & \mu(B_{\frac{\tau}{2\Gamma}(\mM)})\\
    &=& \mu(B_{\frac{\tau}{2\Gamma}(\mH')}) \\
    &\geq & N(\mathcal{O}_\Lambda, \|\cdot\|_2, \frac{\tau}{2\Gamma})^{-1}.
\end{eqnarray*}
Using Theorem \ref{thm:exponential_mechanism_utility}, with query function $q(\mM,\mH) = \inner{\mM}{\mH}$, we have
\begin{eqnarray*}
    \PR\left[\inner{\mM}{\mH} \leq \sum_{i=1}^k\lambda_i\gamma_i -\tau\right] 
    &=&\PR\left[\mH \not\in S_{\tau}\right]\\
    &\leq &\frac{\exp\left(-\frac{\eps}{4\lambda_1}\tau\right)}{\mu\left(S_{\frac{\tau}{2}}\right)}\\
   &\leq & N(\mathcal{O}_\Lambda, \|\cdot\|_2, \frac{\tau}{2\Gamma}) \times \exp\left(-\frac{\eps\tau}{4\lambda_1}\right).
\end{eqnarray*} 
\end{proof}
This gives the following alternative form of the utility bound.
%

\begin{lem}[\textbf{Utility bound for Algorithm \ref{alg:dp_optim_HCIZ}}]
\label{lem:utility_guarantee_alternative}
Let the input and output be as listed in Algorithm \ref{alg:dp_optim_HCIZ}.
For any $\beta \in (0,1)$, with probability at least $1-\beta$, the randomized algorithm $\gM$ in Algorithm \ref{alg:dp_optim_HCIZ} 
outputs a matrix $\mH \in \gO_\mLambda$ satisfying  \begin{equation*}
  \sum_{i=1}^k\gamma_i\lambda_i -   \inner{\mM}{\mH} \leq O\left(\frac{\lambda_1}{\eps}\left(dk\log\sum_{i=1}^d\gamma_i+\log\frac{1}{\beta}\right)\right).
\end{equation*}
\end{lem}
\begin{proof}
We can choose a suitable $\tau$ to give a utility bound on $\inner{\mM}{\mH}$. By letting \begin{equation*}
    \tau = \frac{2\lambda_1}{\eps}\log\left(e+\frac{(2+8\Gamma)^{4dk}}{\beta}\right)+\lambda_1,
\end{equation*}
since $\eps \in (0,1)$,
we have \begin{equation*}
    \tau > \frac{2\lambda_1}{\eps} > \lambda_1.
\end{equation*}
Thus, \begin{equation}
\label{eqn:HCIZ_sampling_utility_1}
    \frac{8\lambda_1\Gamma}{\tau} \leq 8\Gamma.
\end{equation}
Since $\mH$ is sampled from a distribution which is $\frac{\epsilon}{4}$-close to the distribution $\exp(\frac{\eps}{4\lambda_1}\inner{\mM}{\mH})$ in infinity divergence, by plugging in this choice of $\tau$ together with the covering number bound of Lemma \ref{lem:covering_number} into \eqref{eqn:probability_HCIZ_sampling}, we have
\begin{align*}
	\PR\left[\inner{\mM}{\mH} \leq \sum_{i=1}^k\lambda_i\gamma_i - \tau\right]
       &\leq \exp\left(\frac{\eps}{4}\right) \times \exp\left(-\frac{\eps\tau}{4\lambda_1}\right)  \times N\left(\mathcal{O}_\Lambda, \, \, \|\cdot\|_2, \,\, \frac{\tau}{2\sum_{i=1}^d\gamma_i}\right)     \tag{From \eqref{eqn:probability_HCIZ_sampling}}\\
    &\leq \exp\left(-\frac{\eps(\tau-\lambda_1)}{4\lambda_1}\right)\left(1+\frac{16\lambda_1\Gamma}{\tau}\right)^{2dk}    \tag{From Lemma \ref{lem:covering_number}}
\end{align*}
We then substitute $\tau$,
\begin{align*}
    &\hspace{-15mm}
    \PR\left[\inner{\mM}{\mH} \leq \sum_{i=1}^k\lambda_i\gamma_i -O\left(\frac{\lambda_1}{\eps}\left(dk\log\Gamma+\log\frac{1}{\beta}\right) \right)\right]
    \tag{Substitute $\tau$}\\
    &\leq \PR\left[\inner{\mM}{\mH} \leq \sum_{i=1}^k\lambda_i\gamma_i - \tau\right]\\
    &\leq \exp\left(-\frac{\eps(\tau-\lambda_1)}{4\lambda_1}\right)\left(1+\frac{16\lambda_1\Gamma}{\tau}\right)^{2dk}\\
    &\leq \exp\left(-\frac{\eps(\tau-\lambda_1)}{4\lambda_1}\right)\left(1+16\Gamma\right)^{2dk}
    \tag{From \eqref{eqn:HCIZ_sampling_utility_1}}\\
    &= \exp\left(-\log\left(e+\frac{(1+16\Gamma)^{4dk}}{\beta}\right)\right)\left(1+16\Gamma\right)^{2dk}
    \tag{Substitute $\tau$}\\
    & = \left(e+\frac{(1+16\Gamma)^{4dk}}{\beta}\right)^{-1}\left(1+16\Gamma\right)^{2dk}\\
    &\leq \frac{\beta}{(1+16\Gamma)^{4dk}}\left(1+16\Gamma\right)^{2dk}\\
    &= \beta.
\end{align*}
Thus, with probability at least $1-\beta$, we have\begin{equation*}
    \inner{\mM}{\mH} \geq \sum_{i=1}^k\lambda_i\gamma_i -O\left(\frac{\lambda_1}{\eps}\left(dk\log\Gamma+\log\frac{1}{\beta}\right) \right).
\end{equation*}
\end{proof}
\subsection{Running time}
\begin{lem}[\textbf{Running time for Algorithm \ref{alg:dp_optim_HCIZ}}]
\label{lem:computation_time}
The number of arithmetic operations required by the algorithm Algorithm \ref{alg:dp_optim_HCIZ} is polynomial in $\log \frac{1}{\epsilon}$, $\lambda_1$, $\gamma_1 -\gamma_d$, and the number of bits representing $\vlambda=(\lambda_1, \lambda_2, \ldots, \lambda_k)$ and $\vgamma=(\gamma_1, \gamma_2, \ldots, \gamma_d)$.
\end{lem}
\begin{proof}
	This follows directly by using the algorithm  from Corollary 2.7 of \cite{MV21} in Algorithm \ref{alg:dp_optim_HCIZ}.	%
\end{proof}

\subsection{Completing the proof of Theorem \ref{thm:dp_optim_intro}}

\begin{proof}{\bf (of Theorem \ref{thm:dp_optim_intro})}
The privacy Guarantee  for Algorithm \ref{alg:dp_optim_HCIZ} is provided in Lemma 
\ref{lem:dp_optim_privacy_guarantee}.
The utility guarantee is Lemma \ref{lem:utility_guarantee_alternative}.
The running time bound for Algorithm \ref{alg:dp_optim_HCIZ} is from Lemma \ref{lem:computation_time}.

\end{proof}

\section{Differentially Private Rank-$k$ Approximation: Proof of Theorem \ref{thm:dp_rank_k}}\label{Appendix_covariance_estimation}

Our algorithm in the proof of Theorem \ref{thm:dp_rank_k}  has two parts.
The first part approximates the eigenvalues of $M$ and the second part is just Algorithm \ref{alg:dp_optim_HCIZ}.

\subsection{Algorithm}
\begin{algorithm}[ht]
\caption{Differentially private rank-$k$ approximation}
\label{alg:dp_rank_k}
\SetKwInOut{KwInput}{Input}
\SetKwInOut{KwOutput}{Output}
\SetKwInOut{KwAlgorithm}{Algorithm}
\KwInput{A data matrix $\mM \in \gH_+^d \subset \C^{d\times d}$, the rank of output matrix $k\in [d]$, a privacy budget $\epsilon>0$}
\KwOutput{A matrix $\mH \in\gH_+^d \subset \C^{d\times d}$}
\KwAlgorithm{}
\begin{enumerate}
\item Compute the eigenvalues of $\mM$ and let them be $\lambda_1 \geq \cdots \geq \lambda_d$.\\
\item Compute $\Tilde{\lambda}_i \gets \lambda_i + \Lap\left(\frac{4}{\eps}\right)$, for all $i\in [k]$\\
\item Sort $\Tilde{\lambda}_i$s so that $\Tilde{\lambda}_1 \geq \cdots \geq \Tilde{\lambda}_k$\\
\item Define $\mLambda \gets \diag(\Tilde{\lambda}_1, \ldots, \Tilde{\lambda}_k, 0, \ldots, 0) \in \C^{d\times d}$\\
\item Sample $\mH \in \gO_{\mLambda}$ from a distribution that is $\frac{\epsilon}{8}$-close in infinity divergence distance to the \\ distribution $d\nu(\mH) \propto \exp\left(\frac{\eps}{8\lambda_1}\inner{\mM}{\mH}\right)d\mu(\mH)$ 
\item Output $\mH$ and the list of estimated eigenvalues $\Tilde{\lambda}_1, \ldots, \Tilde{\lambda}_k$
\end{enumerate}
\end{algorithm}

\noindent
This algorithm has two parts.
The first part (Step 1 to 2) approximates eigenvalues and is shown to be  $\frac{\epsilon}{2}$-differentially private in Theorem \ref{thm:eigenvalue_approx_proof}.
The second part (Step 3 to 6) is Algorithm \ref{alg:dp_optim_HCIZ} with privacy budget $\frac{\eps}{2}$.

\subsection{First part: Differentially private eigenvalue approximation}
\begin{thm}[\textbf{Differentially private approximation of eigenvalues}]
\label{thm:eigenvalue_approx_proof}
	Given a positive semidefinite (PSD) Hermitian input matrix $\mM \in \gH_+^d$ and a privacy budget $\eps > 0$.
	Let the eigenvalues of $\mM$ be $\lambda_1, \ldots, \lambda_d \in \R$.
	Outputting $\Tilde{\lambda}_1, \ldots, \Tilde{\lambda}_d$, where $\Tilde{\lambda}_i = \lambda_i + \Lap\left(\frac{2}{\eps}\right)$ is an $\eps$-differentially private algorithm for approximating the eigenvalues of $\mM$.
	In addition, for any $i\in[d]$, $\E[\Tilde{\lambda}_i]=\lambda_i$.
	With probability at least $1-\beta$, $|\Tilde{\lambda}_i-\lambda_i| = O\left(\frac{1}{\eps}\log \frac{1}{\beta}\right)$ for all $i$.
\end{thm}

\noindent
Since we need to deal with the eigenvalues, we use the following notation: For any matrix $\mM \in \C^{d\times d}$, we denote $\vlambda(\mM) \coloneqq \left(\lambda_1(\mM), \ldots, \lambda_d(\mM)\right)$ as its eigenvalues with $\lambda_1(\mM)\geq \cdots \geq \lambda_d(\mM)$.
To prove Theorem \ref{thm:eigenvalue_approx_proof}, we need  the following lemma.
\begin{lem}[\textbf{Inequality for eigenvalues}] 
\label{lem:eigenvalue_size}
Given a positive semi-definite Hermitian matrix $\mM \in \gH_+^d \subset \C^{d\times d}$ and a vector $\vv\in \C^d$. Let $\mA \coloneqq \mM-\vv\vv^*$. For any $i\in [d]$, $\lambda_i(\mA) \leq \lambda_i(\mM)$. In addition, $\|\vlambda(\mM)-\vlambda(\mA)\|_1 = \|\vv\|_2^2$.
\end{lem}
\begin{proof}
Let $\sS^{d-1}$ be the sphere of unit vectors in $\C^d$.
For any $\vu \in \sS^{d-1}$, we have
\begin{equation*}
    \vu^*\mA\vu = \vu^*\mM\vu -\vu^*\vv\vv^*\vu = \vu^*\mM\vu- (\vv^*\vu)^2 \leq \vu^*\mM\vu.
\end{equation*}
Thus, pick any $i \in [d]$ and any subspace $U \subseteq \C^d$ with dimension $i$, we have
\begin{equation}
\label{eqn:min_max_eigenvalue_size}
    \min_{\vu\in U \cap \sS^{d-1}}  \vu^*\mA\vu \leq \min_{\vu\in U \cap \sS^{d-1}} \vu^*\mM\vu.
\end{equation}
Thus, using the min-max theorem (Courant–Fischer–Weyl min-max principle), we have
\begin{align*}
    \lambda_i(\mA)\ &=\ \max_{\substack{U \subseteq \C^{d\times d}\colon \dim(U)=i}}\ \min_{\vu\in U \cap \sS^{d-1}}  \vu^*\mA\vu\\
    &\leq\  \max_{\substack{U \subseteq \C^{d\times d}\colon\dim(U)=i}}\ \min_{\vu\in U \cap \sS^{d-1}}\vu^*\mM\vu
    \tag{\eqref{eqn:min_max_eigenvalue_size} holds for any $U \in \C^{d\times d}$}\\
    &=\ \lambda_i(\mM).
\end{align*}
This leads to $\lambda_i(\mA) \leq \lambda_i(\mM)$ for any $i \in [d]$.
In addition,
\begin{align*}
    \|\vlambda(\mM)-\vlambda(\mA)\|_1
    \ &=\ \sum_{i=1}^d |\lambda_i(\mM)-\lambda_i(\mA)|\\ 
    &=\ \sum_{i=1}^d\left(\lambda_i(\mM)-\lambda_i(\mA)\right)\tag{$\lambda_i(\mA)\leq \lambda_i(\mM)$} \\
    &=\ \Tr(\mM)-\Tr(\mA)\\ 
    &=\ \Tr(\vv\vv^*)\\
    &=\ \|\vv\|_2^2.
\end{align*}
\end{proof}
\noindent
Using this lemma, we can then prove Theorem \ref{thm:eigenvalue_approx_proof}.
\begin{proof}{\bf (of Theorem \ref{thm:eigenvalue_approx_proof})}
	Given two neighboring PSD Hermitian matrix $\mM, \mM' \in \gH_+^d \subset \C^{d\times d}$, so that there exist $\vu, \vv \in \C^d$ with $\|\vu\|_2, \|\vv\|_2\leq 1$ such that
$\mM'=\mM-\vu\vu^*+\vv\vv^*$.
Let $\mA \coloneqq \mM-\vu\vu^*$.
Using Lemma \ref{lem:eigenvalue_size}, we have \begin{equation*}
    \|\vlambda(\mM)-\vlambda(\mA)\|_1 = \|\vu\|_2^2 \leq 1.
\end{equation*}
Similarly, we have \begin{equation*}
    \|\vlambda(\mM')-\vlambda(\mA)\|_1 = \|\vv\|_2^2 \leq 1.
\end{equation*}
Thus,
\begin{equation}
    \|\vlambda(\mM)-\vlambda(\mM')\|_1 \leq \|\vlambda(\mM)-\vlambda(\mA)\|_1+\|\vlambda(\mM')-\vlambda(\mA)\|_1 \leq 2.
\end{equation}
Thus, for any neighboring PSD Hermitian matrix $\mM, \mM' \in \gH_+^d$, the $\ell_1$ distance between their eigenvalue vector is at most 2.
According to Definition \ref{def:sensitivity}, the sensitivity of the eigenvalue computation is 2.
Thus, outputting $\Tilde{\lambda}_i = \lambda_i+\Lap\left(\frac{2}{\epsilon}\right)$ follows exactly the Laplace mechanism in Theorem \ref{thm:laplace}.
Thus, the eigenvalue approximation satisfies $\eps$-differential privacy.

In addition, for any $\beta \in (0,1)$, with probability $1-\beta$,
\begin{equation*}
	\left|\Lap\left(\frac{2}{\epsilon}\right)\right| \leq \frac{2}{\epsilon}\ln \frac{1}{\beta} = O\left(\frac{1}{\eps}\log \frac{1}{\beta}\right).
\end{equation*}
Thus, with probability at least $1-\beta$, 
\begin{equation*}
	\left|\Tilde{\lambda}_i-\lambda_i\right| =O\left(\frac{1}{\eps}\log \frac{1}{\beta}\right), 
\end{equation*}
for all $i \in [d]$.
\end{proof}

\subsection{Second part: Completing the proof}

Combining Theorem \ref{thm:eigenvalue_approx_proof} and Theorem \ref{thm:dp_optim_intro}, we can prove Theorem \ref{thm:dp_rank_k}.
\begin{proof}{\bf (of Theorem \ref{thm:dp_rank_k})}
$\;$
\medskip
\newline
\textbf{Running time:}
From Theorem \ref{thm:eigenvalue_approx_proof}, the number of arithmetic operations required by the first part of Theorem \ref{alg:dp_rank_k} (eigenvalue approximations) is $\tilde{O}(d)$.
From Theorem \ref{thm:dp_optim_intro},  the number of arithmetic operations required by the second part of the algorithm is polynomial in $\log \frac{1}{\epsilon}$, $\lambda_1$, and the number of bits representing $\vlambda=(\lambda_1, \lambda_2, \ldots, \lambda_d)$.

\medskip
\noindent
\textbf{Privacy guarantee:}
The first part (eigenvalue approximation) in Algorithm \ref{alg:dp_rank_k} is done by letting each $\Tilde{\lambda}_i \coloneqq \lambda_i + \Lap\left(\frac{4}{\epsilon}\right)$.
Theorem \ref{thm:eigenvalue_approx_proof} implies that this approximation is $\frac{\epsilon}{2}$-differentially private.
The second part of Algorithm \ref{alg:dp_rank_k} is just Algorithm \ref{alg:dp_optim_HCIZ} where we set the privacy budget to be $\frac{\eps}{2}$.
Thus, from Lemma \ref{lem:dp_optim_privacy_guarantee}, the second part is $\frac{\epsilon}{2}$-differentially private.
From the composition theorem of differential privacy \cite{dwork2014algorithmic}, it follows that Algorithm \ref{alg:dp_rank_k} is $\eps$-differentially private.

\medskip
\noindent
\textbf{Utility bound:}
From Theorem \ref{thm:eigenvalue_approx_proof}, it follows that for any $\beta \in (0,1)$, with probability at least $1-\beta$, for all $i \in [k]$, we have 
$
    |\Tilde{\lambda}_i - \lambda_i| \leq O\left(\frac{1}{\eps}\log \frac{1}{\beta}\right).
    $
Note that the values $\tilde{\lambda}=(\tilde{\lambda}_1,\dots,\tilde{\lambda}_k)$ may not be sorted.
Let $\hat{\lambda}$ be the vector generated by sorting the entries of $\tilde{\lambda}$ in non-increasing order. 
For the second part,
we  
note that Theorem \ref{thm:dp_optim_intro} (applied with $\hat{\lambda}_i$s for $\gamma_j$s) implies that, for any $\beta \in (0,1)$, with probability at least $1-\beta$, 
\begin{equation}
\label{eqn:dp_optim_rank_k_inner_prod}
	\inner{\mM}{\mH} \geq \sum_{i=1}^k  \lambda_i\hat{\lambda}_i - O\left(\frac{\hat{\lambda}_1}{\eps}\left(dk\log\sum_{i=1}^d\lambda_i+\log\frac{1}{\beta}\right)\right)
\end{equation} 
Thus, with probability at least $1-\beta$,
\begin{align*}
     \|\mM-\mH\|_F^2 
     \ &=\ \|\mM\|_F^2 + \|\mH\|_F^2 - 2\inner{\mM}{\mH}\\
    &\leq\ \sum_{i=1}^d\lambda_i^2 + \sum_{i=1}^k\hat{\lambda}_i^2 - 2\sum_{i=1}^k \lambda_i\hat{\lambda}_i +  2O\left(\frac{\hat{\lambda}_1}{\eps}\left(dk\log\sum_{i=1}^d\lambda_i+\log\frac{1}{\beta}\right)\right)\\
  	&=\ \sum_{i=k+1}^d \lambda_i^2 + \sum_{i=1}^k\left(\lambda_i^2+\hat{\lambda}_i^2-2\sum_{i=1}^k \lambda_i\hat{\lambda}_i\right)+O\left(\frac{\hat{\lambda}_1}{\eps}\left(dk\log\sum_{i=1}^d\lambda_i+\log\frac{1}{\beta}\right)\right)\\
  	&=\ \sum_{i=k+1}^d \lambda_i^2 + \sum_{i=1}^k\left(\lambda_i-\hat{\lambda}_i\right)^2 +O\left(\frac{\hat{\lambda}_1}{\eps}\left(dk\log\sum_{i=1}^d\lambda_i+\log\frac{1}{\beta}\right)\right).\yesnum\label{eq:before_sorting_argument}
  \end{align*}
  We would like to replace $\hat{\lambda}_1$ in the last term by ${\lambda}_1$ and use Theorem \ref{thm:eigenvalue_approx_proof} to prove an upper bound on the second term.
  For the former, observe that:
  Since with probability at least $1-\beta$, $\left|\Tilde{\lambda}_i-\lambda_i\right| =O\left(\frac{1}{\eps}\log \frac{1}{\beta}\right)$ for all $i$, and $\hat{\lambda}_1=\max_{i}\Tilde{\lambda}_i$,  
  it follows that with probability at least $1-\beta$, $\hat{\lambda}_1=\lambda_1+O\left(\frac{1}{\eps}\log \frac{1}{\beta}\right)$.
  For the latter, notice that for any vector $v\in \R^k$, $\sum_{i=1}^k\left(\lambda_i-v_i\right)^2$ is minimized when the entries in $v$ are sorted in non-increasing order.
  Using these, along with \eqref{eq:before_sorting_argument}, we get that 
  \begin{eqnarray*}
  	\|\mM-\mH\|_F^2  
  	&\leq &\sum_{i=k+1}^d \lambda_i^2 + \sum_{i=1}^k\left(\lambda_i-\Tilde{\lambda}_i\right)^2 +O\left(\frac{{\lambda}_1+\textcolor{black}{\frac{1}{\eps}\log \frac{1}{\beta}}}{\eps}\left(dk\log\sum_{i=1}^d\lambda_i+\log\frac{1}{\beta}\right)\right)\\
  	&\leq &\sum_{i=k+1}^d\lambda_i^2 + kO\left(\frac{1}{\epsilon^2}\log^2\frac{1}{\beta}\right)+O\left(\frac{{\lambda}_1+\textcolor{black}{\frac{1}{\eps}\log \frac{1}{\beta}}}{\eps}\left(dk\log\sum_{i=1}^d\lambda_i+\log\frac{1}{\beta}\right)\right)\\
  	&=&\sum_{i=k+1}^d\lambda_i^2 + O\left(\frac{k}{\epsilon^2}\log^2\frac{1}{\beta}+\frac{{\lambda}_1+\textcolor{black}{\frac{1}{\eps}\log \frac{1}{\beta}}}{\eps}\left(dk\log\sum_{i=1}^d\lambda_i+\log\frac{1}{\beta}\right)\right)\\
  	&=&\sum_{i=k+1}^d\lambda_i^2 + \Tilde{O}\left(\frac{k}{\eps^2}+\frac{dk}{\eps}\left(\lambda_1+\frac{1}{\eps}\right)\right),
  	\end{eqnarray*}
where $\Tilde{O}$ hides logarithmic factors of $\sum_{i=1}^d\lambda_i$ and $\frac{1}{\beta}$.
The above equation uses the fact that with probability at least $1-\beta$, $\left|\Tilde{\lambda}_i-\lambda_i\right| =O\left(\frac{1}{\eps}\log \frac{1}{\beta}\right)$ for all $i$.
Thus, we have proved the utility bounds for Algorithm \ref{alg:dp_rank_k}.

Combining the running time, the privacy guarantee, and the utility bound, we have proved Theorem \ref{thm:dp_rank_k}.
\end{proof}

\section{Packing Number Lower Bound: Proof of Theorem \ref{lemma_orbit_packing_improved}}\label{Appendix_packing_number}

We will use the following result from \cite{szarek1982nets} 
which bounds the covering number $N(\mathcal{G}_{d,k}, \zeta)$ for the (complex) Grassmannian $\mathcal{G}_{d,k}$ with respect to the metric induced by the operator norm on the projection matrices $P_V$ for the subspaces $V \in \mathcal{G}_{d,k}$ (see Proposition 8 of \cite{szarek1998metric}, and the note about about extension to complex spaces).

\begin{lem}[\textbf{Covering number of (complex) Grassmannian $G_{d,k}$ \cite{szarek1982nets, szarek1998metric}}] \label{lemma_covering_Grassmanian}
There exist universal constants $C>c>0$ such that for every unitarily invariant norm $\vvvert  \cdot \vvvert $ and every $0<\zeta <  D_{\vvvert  \cdot\vvvert }$, the covering number $N(\mathcal{G}_{d,k}, \vvvert  \cdot \vvvert , \zeta)$ of the (complex) Grassmannian $G_{d,k}$ satisfies
$$\left(\frac{c D_{\vvvert  \cdot\vvvert }(\mathcal{G}_{d,k})}{\zeta}\right)^{2k(d-k)} \leq N(\mathcal{G}_{d,k}, \vvvert  \cdot \vvvert , \zeta) \leq \left(\frac{C D_{\vvvert  \cdot\vvvert }(\mathcal{G}_{d,k})}{\zeta}\right)^{2k(d-k)},$$
where $D_{\vvvert  \cdot\vvvert }(\mathcal{G}_{d,k}) := \sup_{\mathcal{U}, \mathcal{V} \in \mathcal{G}_{d,k}} \vvvert P_{\mathcal{U}} - P_{\mathcal{V}}\vvvert $ is the diameter of $\mathcal{G}_{d,k}$ with respect to $\vvvert  \cdot \vvvert $.
\end{lem}
In the case of the Frobenius norm, we have $D_{\| \cdot\|_F} = \sup_{\mathcal{U}, \mathcal{V} \in G_{d,k}} \|P_{\mathcal{U}} - P_{\mathcal{V}}\|_F \geq \min(\sqrt{k}, \sqrt{d-k})$,  and in the case of the operator norm we have $D_{\| \cdot\|_2} = \sup_{\mathcal{U}, \mathcal{V} \in G_{d,k}} \|P_{\mathcal{U}} - P_{\mathcal{V}}\|_2 \geq 1$

\noindent
We will also make use of the following Sin-$\Theta$ theorem of \cite{davis1970rotation}.
Let A and $\hat{A}$ be two Hermitian matrices with eigenvalue decompositions
\begin{equation}\label{eq_eigenvalue_decomposition1}
    A = U \Lambda U^\ast = (U_1, U_2) \left({\begin{array}{cc}
   \Lambda_1 &  \\
    &  \Lambda_2 \\
  \end{array}}  \right) \left({\begin{array}{c}
   U_1^\ast \\
      U_2^\ast \\
  \end{array}}  \right)
\end{equation}

\begin{equation}\label{eq_eigenvalue_decomposition2}
    \hat{A} = \hat{U} \hat{\Lambda} \hat{U}^\ast = (\hat{U}_1, \hat{U}_2) \left({\begin{array}{cc}
   \hat{\Lambda}_1 &  \\
    &   \hat{\Lambda}_2 \\
  \end{array}}  \right) \left({\begin{array}{c}
    \hat{U}_1^\ast \\
       \hat{U}_2^\ast \\
  \end{array}}  \right),
\end{equation}
(although when we apply the Sin-Theta theorem we will only need the special case where $\hat{\Lambda} = \Lambda$).

\begin{lem}[\bf sin-$\Theta$ Theorem \cite{davis1970rotation}] \label{lemma_SinTheta}
Let $A, \hat{A}$ be two Hermitian matrices with eigenvalue decompositions given in \eqref{eq_eigenvalue_decomposition1} and \eqref{eq_eigenvalue_decomposition2}.
Suppose that there are $\alpha > \beta>0$ and $\Delta>0$ such that the spectrum of $\Lambda_1$ is contained in the interval $[\alpha, \beta]$ and the spectrum of $\hat{\Lambda}_2$ lies entirely outside of the interval $(\alpha -\Delta, \beta + \Delta)$.
Then
\begin{equation*}
    \vvvert U_1 U_1^\ast - \hat{U}_1 \hat{U}_1^\ast\vvvert  \leq \frac{\vvvert \hat{A}- A\vvvert }{\Delta},
\end{equation*}
where $\vvvert \cdot\vvvert $ denotes the operator norm or Frobenius norm (or, more generally, any  unitarily invariant norm).
\end{lem}

\begin{lem} \label{lemma_unitary_to_projection}
Suppose that $I_k$ is the $d\times d$ diagonal matrix with the first $k$ diagonal entries $1$ and the remaining $d-k$ diagonal entries $0$, and let $\hat{I}_k$ be the first $k$ columns of $I_k$.
Let $P$ be any Hermitian rank-$k$ projection matrix.
Then there exists a $d\times k$ matrix $\hat{W}$ with orthonormal columns such that $\hat{W} \hat{W}^\ast = P$ and  $\| \hat{W} - \hat{I}_k \|_F \leq \| P - I_k \|_F$.
\end{lem}

\begin{proof}
Denote by $\mathcal{I}_k$ the column space of $I_k$ and $\mathcal{P}$ the column space of $P$.
Let $\theta_1 \leq \cdots \leq \theta_k$ be the $k$ principal angles between $\mathcal{I}_k$ and $\mathcal{P}$.
Let $u_1, \ldots, u_k$ and $v_1, \ldots, v_k$, form an orthonormal basis for $\mathcal{I}_k$ and $\mathcal{P}$ respectively, and where the angles between corresponding vectors $u_i$ and $v_i$ in the two bases are equal to the $i$'th principle angle $\theta_i$ for every $i \in [k]$.
The existence of such a basis is guaranteed by the the variational definition of principle angles between subspaces (see e.g.  \cite{bjorck1973numerical}).

Let $\mathcal{I}_k^\perp$ and $\mathcal{P}^\perp$ be the orthogonal complements of $\mathcal{I}_k$ and $\mathcal{P}$, respectively. 
Let $u_{k+1}, \ldots, u_d$ be a basis for  $\mathcal{I}_k^\perp$, and let  $v_{k+1}, \ldots, v_d$ be a basis for  $\mathcal{P}^\perp$.
Let $U_1 =[u_1, \ldots, u_k]$ and $V_1 =[v_1, \ldots, v_k]$.
And let $U_2 =[u_{k+1}, \ldots, u_d]$ and  $V_2 =[v_{k+1}, \ldots, v_d]$
Let $U =[U_1, U_2]$ and $V=[V_1, V_2]$.
Therefore, we have that
\begin{align*} \label{eq_t1}
    \|U_1-V_1\|_F^2 &= \mathrm{tr}((U_1-V_1)^\ast(U_1-V_1))\\
    &= \mathrm{tr}((U_1^\ast U_1 - U_1^\ast V_1 - V_1^\ast U_1 +V_1^\ast V_1)\\
    &\mathrm{tr}(U_1^\ast U_1) -2 \mathrm{tr}( U_1^\ast V_1 ) + \mathrm{tr}(V_1^\ast V_1)\\
    &= 2k -2 \mathrm{tr}( U_1^\ast V_1 )\\
        &= 2k -2\sum_{i=1}^k u_i^\ast v_i\\
    &= 2k -2\sum_{i=1}^k \cos(\theta_i). \yesnum
\end{align*}
But, by the variational definition of principal angles, we also have that the largest singular values of $I_k P^\ast$ are also $\cos(\theta_1)\geq \cos(\theta_2)\geq \cdots \geq \cos(\theta_{k})$, with the remaining singular values equal to $0$.

Therefore, we have that
\begin{align*} \label{eq_t2}
    \|V_1V_1^\ast - I_k \|_F^2 &= \|V_1 V_1^\ast - U_1 U_1^\ast\|_F^2\\
    &= \mathrm{tr}((V_1 V_1^\ast - U_1 U_1^\ast)^\ast(V_1 V_1^\ast - U_1 U_1^\ast))\\
    &= \mathrm{tr}(V_1 V_1^\ast)^2) - \mathrm{tr}((U_1 U_1^\ast)(V_1 V_1^\ast))  - \mathrm{tr}((V_1 V_1^\ast)(U_1 U_1^\ast)) + \mathrm{tr}((U_1 U_1^\ast)^2)\\
    & = 2k  - 2\mathrm{tr}((U_1 U_1^\ast)(V_1 V_1^\ast))\\
    & = 2k  - 2\mathrm{tr}(I_k P^\ast)\\
    & \geq 2k - 2\sum_{i=1}^k \cos(\theta_i), \yesnum
\end{align*}
where the inequality holds since $\mathrm{tr}((U_1 U_1^\ast)(V_1 V_1^\ast))$ is the sum of the eigenvalues of $(U_1 U_1^\ast)(V_1 V_1^\ast)$, and the sum of the singular values of any matrix is at least as large as the sum of its eigenvalues.
Therefore, combining \eqref{eq_t1} a \eqref{eq_t2} we have that
\begin{equation}\label{eq_t3}
 \|U_1-V_1\|_F \leq \|V_1V_1^\ast - I_k \|_F.
\end{equation}
But, since $V_1= V_1 U_1^\ast U_1$ we also have that
\begin{align*}\label{eq_t4}
    \|V_1-U_1\|_F^2 &=  \|V_1 U_1^\ast U_1 - U_1 U_1^\ast U_1\|_F\\
    &= \|(V_1 U_1^\ast - I_k)U_1\|_F^2\\
    & = \mathrm{tr}(((V_1 U_1^\ast - I_k)U_1)^\ast (V_1 U_1^\ast - I_k)U_1)\\
    & = \mathrm{tr}( U_1^\ast (V_1 U_1^\ast - I_k)^\ast (V_1 U_1^\ast - I_k) U_1)\\
    & = \mathrm{tr}((V_1 U_1^\ast - I_k)^\ast (V_1 U_1^\ast - I_k) U_1 U_1^\ast)\\
  & = \mathrm{tr}((V_1 U_1^\ast - I_k)^\ast (V_1 U_1^\ast - I_k))\\
  &= \|V_1 U_1^\ast - I_k\|_F^2\\
    &= \|V_1 U_1^\ast  U_1 U_1^\ast - I_k\|_F^2\\
        &= \|V_1 U_1^\ast  I_k - I_k\|_F^2\\
                &= \|V_1 U_1^\ast  \hat{I}_k - \hat{I}_k\|_F^2. \yesnum
\end{align*}
Therefore, combining \eqref{eq_t3} and \eqref{eq_t4} we have that
\begin{equation}\label{eq_t5}
      \|V_1 U_1^\ast \hat{I}_k - \hat{I}_k\|_F  \leq     \|V_1V_1^\ast - I_k \|_F
\end{equation}
Now, since $U$ is a unitary matrix, $V U^\ast$ is also unitary.
But, since $U_1 U_1^\ast = I_k$,  $U_1$ must have all zeros below the $k'th$ row, and $U_2$ must have all zeros above the $k+1$'st row.
Therefore, the first $k$ columns of $V U^\ast$ are the same as the first $k$ columns of $V_1 U_1^\ast$.
Thus, the first $k$ columns of  $V_1 U_1^\ast$ must be orthonormal to each other.
Since the columns of $V_1 U_1^\ast \hat{I}_k$ are the same as the first $k$ columns of $V_1 U_1^\ast$, the columns of $V_1 U_1^\ast \hat{I}_k$ must be orthonormal to each other as well.
Therefore, setting  $\hat{W} = V_1 U_1 ^\ast \hat{I}_k$, and since $V_1V_1^\ast=P$, from \eqref{eq_t5} we have that
\begin{equation*}
  \|\hat{W} - \hat{I}_k\|_F  \leq     \|P - I_k \|_F 
  \end{equation*}
  where  $\hat{W} = V_1 U_1 ^\ast \hat{I}_k$, is a matrix with orthonormal columns. 

The remaining $d-k$ orthonormal columns of the unitary matrix $W$ (with the first $k$ columns being the columns of $\hat{W}$) can be found by diagonalizing the projection matrix for the subspace $\mathcal{P}^\perp$.

\end{proof}

\begin{lem}[\bf Packing number of ball inside $\mathcal{O}_\Lambda$]\label{lemma_packing_subset_ball}
For any $M \in \mathcal{O}_\Lambda$, any unitarily invariant norm $\vvvert  \cdot \vvvert $, and any $\delta>r>0$, and denoting by $B(x, \vvvert  \cdot \vvvert ,\delta)$ a ball of radius $\delta$ centered at $M$ with respect to the norm m $\vvvert  \cdot \vvvert $, we have
\begin{equation*}
  P(B(M, \vvvert  \cdot \vvvert ,\delta) \cap \mathcal{O}_\lambda, \vvvert  \cdot \vvvert , r) \geq \frac{P(\mathcal{O}_\Lambda, \vvvert  \cdot \vvvert , r)}{N(\mathcal{O}_\Lambda, \vvvert  \cdot \vvvert , \delta)}.   
\end{equation*}
\end{lem}
\begin{proof}
Let $M, M' \in \mathcal{O}_\Lambda$.
Then there exists $U \in U(d)$ such that $M' = U M U^\ast$.
Since $\vvvert \cdot\vvvert $ is unitarily invariant, $\vvvert  W_1 - W_2 \vvvert  = \vvvert  U(W_1 - W_2)U^\ast \vvvert  = \vvvert  U W_1 U^\ast - U W_2 U^\ast \vvvert $ for any $W_1, W_2 \in \mathcal{O}_\Lambda$.
 Thus, for any $n \in \mathbb{N}$, we have that a collection of matrices $M_1,\ldots,M_n$ is an  $r$-packing of  $B(M, \vvvert  \cdot \vvvert ,\delta)$ if and only if $UM_1U^\ast,\ldots,UM_nU^\ast$ is an $r$-packing of $B(UMU^\ast, \vvvert  \cdot \vvvert ,\delta)$.
 Thus, for every $M, M' \in  \mathcal{O}_\Lambda$ we have that
 \begin{equation}\label{eq_packing_invariance}
      P(B(M, \vvvert  \cdot \vvvert ,\delta), \vvvert  \cdot \vvvert , r) =  P(B(M', \vvvert  \cdot \vvvert ,\delta), \vvvert  \cdot \vvvert , r).
 \end{equation}
Let $M_1, \ldots, M_\alpha$, where $\alpha \geq P(\mathcal{O}_\Lambda, \vvvert  \cdot \vvvert , r)$, be an $r$-packing of $\mathcal{O}_\Lambda$.
And let $C_1, \ldots, C_\beta$ where $C_1, \ldots, C_\beta$ are balls of radius $\delta$ centered in $\mathcal{O}_\Lambda$ and $\beta \leq N(\mathcal{O}_\Lambda, \vvvert  \cdot \vvvert , \delta)$, be a $\delta$-covering of $\mathcal{O}_\Lambda$.
Then by the pigeonhole principle there exists $i \in [\beta]$ such that a subset of $M_1, \ldots, M_\alpha$ of size $\geq \frac{\alpha}{\beta}$ is an $r$-packing of $C_i$.
Hence, $P(C_i, r) \geq \frac{\alpha}{\beta}$ and by \eqref{eq_packing_invariance} we have that
\begin{align*}
    P(B(M, \vvvert  \cdot \vvvert ,\delta), \vvvert  \cdot \vvvert , r) &= P(C_i, r)\\
    &\geq \frac{\alpha}{\beta}\\
    &\geq \frac{P(\mathcal{O}_\Lambda, \vvvert  \cdot \vvvert , r)}{N(\mathcal{O}_\Lambda, \vvvert  \cdot \vvvert , \delta)}. \yesnum
\end{align*}

\end{proof}

\noindent
For any $\zeta>0$ and norm $\vvvert  \cdot \vvvert $, we define the packing number $P(S, \vvvert  \cdot \vvvert ,\zeta)$ of a subset $S \subset E$ of a normed space $E$ with metric $\rho$ induced by the norm $\vvvert  \cdot \vvvert $ to be the largest subset $\{x_1, \ldots, x_n\}$, for any $n \in \mathbb{N}$, such that $\rho(x_i, x_j) >  \zeta$.

We will now use the covering number for the Grassmannian and the Sin-$\Theta$ theorem to prove Theorem \ref{lemma_orbit_packing_improved}.\\
\begin{proof}{\bf (of Theorem \ref{lemma_orbit_packing_improved})}

\noindent
{\bf Bounding the packing number of a ball in the Grassmannian:}
Denote by $I_k \in \mathcal{P}_{d,k},$ the matrix with its first $k$ diagonal entries $1$ and each all other entries $0$.
Plugging in the upper and lower bounds for the covering number in \ref{lemma_covering_Grassmanian} into Lemma \ref{lemma_packing_subset_ball}, for any $k>0$, we get that the packing number of any ball of radius $\zeta$ inside any Grassmannian manifold $\mathcal{G}_{d,k}$ (which we represent by the set of rank-$k$ projection matrices $\mathcal{P}_{d,k}$), with center $I_k \in \mathcal{P}_{d,k}$ \footnote{Note that the choice of center here is arbitrary, and we would get the same bound regardless of choice of center since $\|\cdot\|_F$ is unitarily invariant.}, satisfies
\begin{align*}\label{packing_subset}
  P(B(I_k, \|  \cdot \|_F ,\omega) \cap \mathcal{P}_{d,k}, \|  \cdot \|_F , \zeta) & \stackrel{\textrm{Lemma \ref{lemma_packing_subset_ball}}}{\geq}  \frac{P(\mathcal{P}_{d,k}, \|  \cdot \|_F , \zeta)}{N(\mathcal{P}_{d,k}, \|  \cdot \|_F, \omega)} \\
    &=\frac{P(\mathcal{P}_{d,k}, \|  \cdot \|_F , \zeta)}{N(\mathcal{P}_{d,k}, \|  \cdot \|_F, \omega)}\\
  &\geq \frac{N(\mathcal{G}_{d,k}, \|  \cdot \|_F , 2\zeta)}{N(\mathcal{G}_{d,k}, \|  \cdot \|_F , \omega)}\\
  &\stackrel{\textrm{Lemma \ref{lemma_covering_Grassmanian}}}{\geq}  \left(\frac{\min(\omega, D_{\| \cdot\|_F }(\mathcal{G}_{d,k}) c }{ 2 \zeta C}\right)^{2k(d-k)}\\
  &\geq \left(\frac{\min(\omega, \sqrt{k}, \sqrt{d-k}) c }{ 2 \zeta C}\right)^{2k(d-k)}. \yesnum
\end{align*}

\noindent
{\bf Constructing the map from the Grassmannian to the orbit:}
For any projection matrix $M \in \mathcal{P}_{d-j+i+1, \, \, i}$ define a map $\psi: \mathcal{P}_{d-j+i+1, \, \, i} \rightarrow \mathcal{U}(d-j+i+1)$, from $\mathcal{P}_{d-j+i+1, \, \, i}$ to the group of $(d-j+i+1) \times (d-j+i+1)$ unitary matrices $\mathcal{U}(d-j+i+1)$, as follows:
\begin{itemize}
    \item $\psi(I_i) = I$, where $I_i$ is the $(d-j+i+1) \times (d-j+i+1)$ diagonal matrix with the first $i$ diagonal entries $1$ and all other entries $0$,  and $I$ is the $(d-j+i+1) \times (d-j+i+1)$ identity matrix.
       \item  $\psi(M) = U$, where $U  \in \mathcal{U}(d-j+i+1)$ is a unitary matrix such that its first $i$ columns $U_1$ satisfy  $U_1 U_1^\ast = M$, and  $\|U - I\|_F \leq  2\|M - I_i \|_F$.
\end{itemize}

\noindent
We still need to show that a matrix $\psi(M) = U$  satisfying the above conditions exists.
We can construct the matrix  $\psi(M) = U$ by applying Lemma \ref{lemma_unitary_to_projection} twice.
First, we apply Lemma \ref{lemma_unitary_to_projection} which guarantees the existence of matrix $U_1$ with orthonormal columns such that $U_1 U_1^\ast = M$ and $\|U_1 - \hat{I}_i\|_F \leq \|M - I_i \|_F$.
Next, we apply Lemma \ref{lemma_unitary_to_projection} a second time to obtain a matrix $U_2$ with orthonormal columns such that $U_2 U_2^\ast = I- M$ is a projection matrix for the orthogonal complement of the space spanned by the columns of $M$, and $\|U_2 - (\hat{I}-\hat{I}_i)\|_F \leq \|(I-M) - (I-I_i) \|_F = \|M - I_i\|_F$.
Define the matrix $U:= [U_1,U_2]$.
Then we have that
\begin{align*}
    \|U - I\|_F &\leq \|U_1 - \hat{I}_i\|_F  + \|U_2 - (\hat{I}-\hat{I}_i)\|_F\\
    &\leq  2\|M - I_i \|_F.
\end{align*}
Moreover, since $U_1 U_1^\ast = M$ and $U_2 U_2^\ast = I- M$, we have that the columns of $U_1$ and $U_2$ are orthogonal to each other and hence that the matrix $U$ is a  $(d-j+i+1) \times (d-j+i+1)$ unitary matrix.

\medskip

\noindent {\bf Showing that the map preserves Frobenius norm distance (lower bound):}
For convenience, we denote the submatrix of any matrix $H$ consisting of the entries in rows $k, \ldots, \ell$ by $H[k : \ell]$.
For convenience, in the remainder of the proof, we denote the restriction of $\psi$ to the first $i$ columns of its output by $\psi_1(M)= U_1$.
And we denote the last $d-j+1$ columns of $U$ by $U_2$, and the restriction of $\psi$ to these columns by $\psi_2(M)= U_2$.

Next, consider the map $\Psi: \mathcal{P}_{i,d-j+i+1} \rightarrow \mathcal{U}(d)$ defined as follows:
\begin{equation*}
    \Psi(M) := \left({\begin{array}{ccc}
   \psi_1(M)[1:i] & 0 &  \psi_2(M)[1:i]\\
    0 & I_{(j-i -1) \times (j-i-1)} & 0 \\
     \psi_1(M)[i+1:d-j+1]  & 0 & \psi_2(M)[i+1:d-j+1] \\
  \end{array}}  \right),
\end{equation*}
where $I_{(j-i-1)\times (j-i-1)}$ denotes the $(j-i-1)\times (j-i-1)$ identity matrix.
And define the map $\phi: \mathcal{P}_{i,d-j+i+1} \rightarrow  \mathcal{O}_{\Lambda}$ as follows:
\begin{equation} \label{eq_phi_map}
   \phi(M) = \Psi(M) \, \Lambda \, \Psi(M)^\ast
\end{equation}
Define $\tilde{\Lambda}:= \mathrm{diag}(\lambda_1, \ldots, \lambda_i, \lambda_{j}, \ldots, \lambda_d).$
For any projection matrices $M, M' \in \mathcal{P}_{i,d-j+i+1}$, we have
\begin{align*} \label{eq_operator_norms_b}
    \vvvert  \phi(M) - \phi(M') \|_F  &=  \vvvert  \Psi(M) \Lambda  \Psi(M)^\ast - \Psi(M') \Lambda  \Psi(M')^\ast \|_F \\
     & = \|  \psi(M) \tilde{\Lambda}  \psi(M)^\ast - \psi(M') \tilde{\Lambda}  \psi(M')^\ast \|_F \\
    &\geq  (\lambda_i - \lambda_j) \times \|  \psi_1(M) \psi_1(M)^\ast - \psi_1(M') \psi_1(M')^\ast \|_F \\
    &= (\lambda_i - \lambda_j) \times \| M - M' \|_F , \yesnum
\end{align*}
where the inequality holds by the Sin-$\Theta$ Theorem of  \cite{davis1970rotation} (restated above as Lemma \ref{lemma_SinTheta}), since $\| \cdot\|_F $ is a unitarily invariant norm.

\medskip
\noindent
{\bf Showing that the map preserves Frobenius norm distance (upper bound):}
Moreover, we also have that
\begin{align*} \label{eq_inside_ball}
        \| \phi(M) - \Lambda\|_F  &= \| \phi(M) - \phi(I_i)\|_F\\
        &=  \|  \Psi(M) \Lambda  \Psi(M)^\ast - \Psi(I_i) \Lambda  \Psi(I_i)^\ast\|_F\\
             & = \|  \psi(M) \tilde{\Lambda}  \psi(M)^\ast - \psi(I_i) \tilde{\Lambda}  \psi(I_i)^\ast\|_F \\
                          & = \|  \psi(M) \tilde{\Lambda}  \psi(M)^\ast - I \tilde{\Lambda}  I^\ast\|_F \\
                          &=\|  (\psi(M) -I)\tilde{\Lambda}  \psi(M)^\ast - I\tilde{\Lambda}  (I^\ast - \psi(M))\|_F\\
                          &\leq \|(\psi(M) -I)\tilde{\Lambda}  \psi(M)^\ast \|_F + \|I\tilde{\Lambda}  (I^\ast - \psi(M))\|_F\\
                          &= 2\|(\psi(M) -I)\tilde{\Lambda}  \psi(M)^\ast \|_F\\
                          &\leq 2\|(\psi(M) -I)\tilde{\Lambda}\|_F \times \|\psi(M)^\ast \|_2\\
                          &\leq 2\|(\psi(M) -I)\tilde{\Lambda}\|_F\\
                           &\leq 2\|\psi(M) -I\|_F \times \|\tilde{\Lambda}\|_2\\
                        &\leq 2\lambda_1\|\psi(M) -I\|_F, \yesnum
\end{align*}
where the second and fourth inequalities hold because the Freobenius norm is sub-multiplicative with respect to the operator norm, and the third inequality holds because $\|\psi(M)^\ast \|_2=1$ since $\psi(M)^\ast$ is a unitary matrix.

\medskip
\noindent {\bf Bounding the packing number (subset of orbit):}
From \eqref{packing_subset}, we have that
\begin{align*}
  P\left(B\left(I_k, \|  \cdot \|_F ,\frac{\omega}{2\lambda_1}\right) \cap \mathcal{P}_{i,d-j+i+1}, \, \, \|  \cdot \|_F , \, \, \frac{\zeta}{\lambda_i - \lambda_j}\right) \geq J,
  \end{align*}
  where
\begin{align*}
 J = \left(\frac{\min(\omega, \lambda_1 \sqrt{i},  \lambda_1 \sqrt{d-j+1} ) c \times (\lambda_i - \lambda_j)}{ 2\lambda_1 \zeta C}\right)^{2i\times (d-j+1)}.
  \end{align*}
  Therefore, we have that there exists a $\{M_1,\ldots M_J\} \subseteq \mathcal{P}_{i,d-j+i+1}$ of size $J$ such that,
  \begin{equation} \label{eq_p1}
    \| M_s - M_t\|_F  > \frac{\zeta}{\lambda_i - \lambda_j} \qquad \forall s,t \in  [J]
\end{equation}
and 
    \begin{equation} \label{eq_p2}
    \| M_s - I_k\|_F < \frac{\omega}{2\lambda_1} \qquad \forall s \in [J].
\end{equation}
  Therefore, plugging by \eqref{eq_operator_norms_b} into that \eqref{eq_p1} we have that,
\begin{equation}  \label{eq_p3}
    \| \phi(M_s) - \phi(M_t)\|_F  > \zeta \qquad \forall s,t \in  [J]
\end{equation}
and, moreover, plugging \eqref{eq_inside_ball} into \eqref{eq_p1} we have that
\begin{equation} \label{eq_p4}
\| \phi(M_s) - \Lambda\|_F < \omega \qquad \forall s \in [J].
\end{equation}
Since by \eqref{eq_phi_map}, $\phi(M_1),\ldots \phi(M_J)$ are all in the unitary orbit $\mathcal{O}_\Omega$,  \eqref{eq_p3} and \eqref{eq_p4} imply that $\phi(M_1),\ldots \phi(M_J)$ is a $\zeta$ packing for $B(\Lambda, \omega) \cap \mathcal{O}_\Lambda$.
Therefore, the packing number of  $B(\Lambda, \omega) \cap \mathcal{O}_\Lambda$ is
\begin{equation} \label{eq_p5}
P(B(\Lambda,\omega)\cap \mathcal{O}_\Lambda, \vvvert \cdot\vvvert , \zeta) \geq J.    
\end{equation}
But since $\|\cdot\|_F$ is unitarily invariant, we have that
\begin{equation} \label{eq_p6}
    P(B(\Lambda,\omega)\cap \mathcal{O}_\Lambda, \| \cdot\|_F , \zeta)  = P(B(X,\omega)\cap \mathcal{O}_\Lambda, \| \cdot\| , \zeta) \qquad \forall  X \in \mathcal{O}_\Lambda. 
\end{equation}
Therefore,  \eqref{eq_p5} and \eqref{eq_p6} together imply that
\begin{equation}  \label{eq_p7}
 P(B(X,\omega)\cap \mathcal{O}_\Lambda, \| \cdot\|_F , \zeta) \leq  \left(\frac{\min(\omega, \lambda_1 \sqrt{i},  \lambda_1 \sqrt{d-j+1} ) c \times (\lambda_i - \lambda_j)}{ 2\lambda_1 \zeta C}\right)^{2i\times (d-j+1)}
 \end{equation}
 for every $1\leq i < j \leq d$.
 Therefore, we have that
\begin{align*} 
 \log P(B(&X,\omega)\cap \mathcal{O}_\Lambda, \| \cdot\|_F , \zeta)\\
 &\geq \max_{1\leq i < j \leq d}  2i\times (d-j+1) \times  \log  \left(\frac{\min(\omega, \lambda_1 \sqrt{i},  \lambda_1 \sqrt{d-j+1} ) c \times (\lambda_i - \lambda_j)}{ 2\lambda_1 \zeta C}\right).
 \end{align*}
 This completes the proof of \eqref{eqref_packing_bound_subset}.

\medskip
\noindent {\bf Bounding the packing number (entire orbit, with slightly stronger bound):}
We can get a slightly better bound when bounding the entire unitary orbit, using the following argument.
Since $$P(\mathcal{P}_{i,d-j+i+1},\| \cdot\|_F , \frac{\zeta}{\lambda_i - \lambda_j}) \geq \left(\frac{cD_{\| \cdot\|_F }(\mathcal{G}_{d-j+i+1,\, \,i})\times(\lambda_i - \lambda_j)}{\zeta}\right)^{2i\times (d-j+1)},$$ there exits a subset $\{M_1,\ldots M_n\} \subseteq \mathcal{P}_{i,d-j+i+1}$ of size $n=\left(\frac{cD_{\vvvert \cdot\vvvert }(\mathcal{G}_{d-j+i+1,\, \,i})\times(\lambda_i - \lambda_j)}{\zeta}\right)^{2i\times (d-j+1)}$ such that $\| M_r - M_s\|_F  > \frac{\zeta}{\lambda_i - \lambda_j}$ for all $r,s \in  [n]$.
Thus, by \eqref{eq_operator_norms_b} we have that
\begin{equation} \label{eq_operator_norms2}
    \|  \phi(M_r) - \phi(M_s)\|_F  > \zeta, \quad \quad \forall r,s \in [n].
\end{equation}
Since we have a subset $\{\phi(M_1),\ldots \phi(M_n)\} \subseteq \mathcal{O}_\Lambda$ such that  $\vvvert  \phi(M_r) - \phi(M_s)\vvvert  > \zeta$ for all $r,s \in [n]$, the packing number of $\mathcal{O}_\Lambda$ satisfies 
\begin{equation} \label{eq_packing}
    P(\mathcal{O}_\Lambda,\| \cdot\|_F , \zeta) \geq n =\left(\frac{cD_{\| \cdot\|_F }(\mathcal{G}_{d-j+i+1,\, \,i})\times(\lambda_i - \lambda_{j})}{\zeta}\right)^{2i\times (d-j+1)}.
    \end{equation}
Finally, since \eqref{eq_packing} holds for every choice of $1\leq i<j\leq d$, and $D_{\| \cdot\|_F }(\mathcal{G}_{d-j+i+1,\, \,i}) \geq \Omega( \min(\sqrt{i}, \sqrt{d-j+1}))$, we have that $$P(\mathcal{O}_\Lambda, \|\cdot \|_F, \zeta) \geq \max_{1\leq i < j \leq d}   \left(\frac{c \min(\sqrt{i}, \sqrt{d-j+1})\times(\lambda_i - \lambda_{j})}{\zeta}\right)^{2i\times (d-j+1)}.$$
%
This completes the proof of \eqref{eqref_packing_bound_entire_orbit}.
\end{proof}

\section{Lower Bound on Utility: Proof of Theorem \ref{thm_lower_bound_general_orbit} }\label{appendix_lower_utility_bound}

\noindent
To prove the lower bound on the utility, we will also use the following lemma:
\begin{lem}\label{lemma_frobenius_inner_produc}
For any $U,V \in \mathcal{U}(d)$ we have
\begin{equation}
    \|U\Lambda U^\ast - V\Lambda V^\ast \|_F^2 = 2 \langle  U \Lambda U^\ast, \, \, U \Lambda U^\ast - V\Lambda V^\ast \rangle \notag
\end{equation}
\end{lem}

\begin{proof}
\begin{align}
   \|U\Lambda U^\ast - V\Lambda V^\ast \|_F^2  &= \langle U \Lambda U^\ast - V\Lambda V^\ast, \, \,  U \Lambda U^\ast - V\Lambda V^\ast \rangle \notag\\
   & = \mathrm{tr}((U \Lambda U^\ast - V\Lambda V^\ast)^\ast (U \Lambda U^\ast - V\Lambda V^\ast)))\notag\\
   & = \mathrm{tr}((U \Lambda U^\ast - V\Lambda V^\ast)^2)\notag\\
      & = \mathrm{tr}(U \Lambda^2 U^\ast - U \Lambda U^\ast V\Lambda V^\ast - V\Lambda V^\ast U \Lambda U^\ast + V\Lambda^2 V^\ast)\notag\\
    & = 2\mathrm{tr}(\Lambda^2) - 2\mathrm{tr}(U \Lambda U^\ast V\Lambda V^\ast)\notag\\
          & = 2\mathrm{tr}(U\Lambda^2U^\ast - U \Lambda U^\ast V\Lambda V^\ast)\notag\\
         & = 2\mathrm{tr}(U\Lambda U^\ast(U\Lambda U^\ast -V\Lambda V^\ast)). \notag
\end{align}
\end{proof}

\begin{lem}[\bf Lower utility bound for unitary orbit, as a function of packing number]\label{thm_lower_bound_orbit}
Suppose, for some $\alpha>\eta >0$, that $\lambda_1 \geq \cdots \geq \lambda_d \geq 0$ and $\epsilon>0$ are such that $$\sum_{\ell=1}^d \lambda_\ell^2 
< \frac{1}{16\eps \delta  \alpha^2} \log(P(B(W, 2\alpha r) \cap \mathcal{O}_\Lambda, \|\cdot\|_F, 2\eta r))) - \frac{d}{16 \delta  \alpha^2} $$
for some $\delta>0$, where we define $r:= \sqrt{\delta \sum_{\ell=1}^d \lambda_\ell^2}$ and $W \in \mathcal{O}_\Lambda$ is any matrix in the unitary orbit.\footnote{The choice of $W$ does not matter since unitary invariance means that the packing bound depend only on the radius of the ball we are packing, not its center.}
Then for any $\epsilon$-differentially private algorithm $\mathcal{A}$ which takes as input a Hermitian matrix and outputs a matrix in the orbit $\mathcal{O}_\Lambda$,  there exists a Hermitian matrix $M$ with eigenvalues $\lambda_1 \geq \cdots \geq \lambda_d$ such that the output  $\mathcal{A}(M)$ of the algorithm satisfies
\begin{equation} \label{eq_bad_utility}
  \sum_{\ell=1}^d \lambda_\ell^2 -  \langle M, \mathcal{A}(M) \rangle \geq \eta^2\delta \sum_{\ell=1}^d \lambda_\ell^2
\end{equation}
with probability at least $\frac{1}{2}$.
\end{lem}

\begin{proof}
Let $W \in \mathcal{O}_\Lambda$ be such that $\langle M, W \rangle = \sum_{\ell=1}^d \lambda_\ell^2$. Define $r:= 2\sqrt{\delta \sum_{i=1}^d \lambda_i^2}$.
By the definition of packing number, there exists a $2\eta r$-packing $E = \{U_i\Lambda U_i^ast\}_{i=1}^n$ of $B(W, 2\alpha r)$ for $n = P(B(W, 2\alpha r), \|\cdot\|_F, 2\eta r))$ such that for every $i,j \in [n]$, $i \neq j$, we have $$B(U_i\Lambda U_i^\ast, \eta r) \cap B(U_j\Lambda U_j^\ast, \eta r) = \emptyset.$$
Consider the matrices $M_i = U_i \Lambda U_i^\ast$ for each $i\in [n]$.
We would like to show that \eqref{eq_bad_utility} holds for one of these matrix $M_i$.
Suppose, on the contrary, that for every $i\in [n]$, the output of the algorithm, $\mathcal{A}(M_i) = V_i \Lambda V_i^\ast$ (where we denote by $V_i \in \mathcal{U}(d)$ a unitary matrix which diagonalizes $M_i$), satisfies
\begin{equation} \label{eq_event}
    \langle M_i, \mathcal{A}(M_i) \rangle > (1-\eta^2 \delta)\sum_{\ell=1}^d \lambda_\ell^2,
\end{equation}
with probability at least $\frac{1}{2}$.
Let $E_i$ be the event that \eqref{eq_event} is satisfied.  The $\mathbb{P}(E) \geq \frac{1}{2}$.

Suppose that the event $E_i$ occurs.
Then, since $\langle M_i, M_i \rangle = \sum_{\ell=1}^d \lambda_\ell^2$, we have that
\begin{equation*}
     \langle M_i, \mathcal{A}(M_i) - M_i \rangle = \langle M_i, \mathcal{A}(M_i) \rangle - \langle M_i, M_i \rangle > -\eta^2 \delta\sum_{\ell=1}^d
     \lambda_\ell^2,
\end{equation*}
Therefore by Lemma \ref{lemma_frobenius_inner_produc} we have that
\begin{align*}
   \delta \sum_{\ell=1}^d \lambda_\ell^2 & > \langle M_i, \, \, M_i - \mathcal{A}(M_i) \rangle\\
   &= \langle  U_i \Lambda U_i^\ast,\, \, U_i \Lambda U_i^\ast -  V_i \Lambda V_i^\ast \rangle\\
   &=  \frac{1}{2} \|U_i \Lambda U_i^\ast -  V_i \Lambda V_i^\ast\|_F^2\\
      &=  \frac{1}{2} \|M_i -  \mathcal{A}(M_i)\|_F^2. \yesnum
\end{align*}
That is,
\begin{equation*}
   \|M_i -  \mathcal{A}(M_i)\|_F^2 < 2 \eta^2 \delta \sum_{\ell=1}^d \lambda_\ell^2.
\end{equation*}
whenever the event $E_i$ occurs.
Thus, since $\mathbb{P}(E_i)\geq \frac{1}{2}$, 
for every $i\in [n]$, we have that
\begin{equation*}
    \mathbb{P}\left(\|M_i -  \mathcal{A}(M_i)\|_F^2 < 2 \eta^2  \delta \sum_{\ell=1}^d \lambda_\ell^2\right) \geq \frac{1}{2}.
\end{equation*}
Hence,  for every $i\in [n]$,
\begin{equation} \label{eq_probability_in_ball}
    \mathbb{P}\left(\|M_i -  \mathcal{A}(M_i)\|_F < 2 \eta \sqrt{  \delta \sum_{\ell=1}^d \lambda_\ell^2}\right) \geq \frac{1}{2}.
\end{equation}
Inequality \ref{eq_probability_in_ball} implies that the output $\mathcal{A}(M_i)$ of Algorithm $\mathcal{A}$ falls inside the Frobenius-norm ball $B\left(M_i, r\right)$ with probability at least $\frac{1}{2}$.

For any $i,j \in [n]$, we have that $M_i - M_j = \sum_{s=1}^m x_s x_s^*$ for some $m \leq \|M_i-M_j\|_F^2 +d$ and data vectors $x_1,\ldots, x_m \in \mathbb{C}^d$ for which $\|x_s\|_2 \leq 1$.
Thus, one can modify the data matrix $M_i$ into any other data matrix $M_j$ by replacing at most $\|M_i - M_j\|_F^2 + d$ points in the dataset.
Since by assumption, Algorithm $\mathcal{A}$ is $\eps$-differentially private, we have that for any $i,j \in [n]$,
\begin{align*} \label{eq_probability_in_ball_2}
     \frac{\mathbb{P}(\mathcal{A}(M_i) \in B(M_j, \eta r))}{\mathbb{P}(\mathcal{A}(M_i) \in B(M_i, \eta r))} &\geq e^{-\eps(\|M_i-M_j\|_F^2 + d)}\\
     &\geq  e^{-\eps(d +16\alpha^2 r^2)}, \yesnum
\end{align*}
since $M_i, M_j \in B(W, 2\alpha r)$.
But by \eqref{eq_probability_in_ball} we have that $\mathbb{P}(\mathcal{A}(M_i) \in B(M_i, \eta r)) > \frac{1}{2}$.
Therefore, \eqref{eq_probability_in_ball_2} implies that
\begin{equation} \label{eq_probability_in_ball_3}
     \mathbb{P}(\mathcal{A}(M_i) \in B(M_j, \eta r)) \geq  \frac{1}{2}e^{-\eps(d +16 \alpha^2 r^2)} \qquad \forall i \in [n].
\end{equation}
Since $M_1,\ldots, M_n$ is a $2\eta r$-packing of $B(W, 2\alpha r)$, the balls $B(M_j, \eta r)$, $j \in [n]$, are pairwise disjoint.
Thus,
\begin{align} \label{eq_sum_of_probabilities}
1 \geq \sum_{j=1}^n \mathbb{P}(\mathcal{A}(M_i) \in B(M_j, \eta r)) \geq n \times e^{-\eps(d +16\alpha^2 r^2)}
\end{align}
Rearranging \eqref{eq_sum_of_probabilities}, we have that
\begin{align*}
\log(n) \leq \eps(d + 16 \alpha^2 r^2)
\end{align*}
and hence that
\begin{align*} 
\frac{1}{16\eps \alpha^2} \log(n) - \frac{d}{16 \alpha^2} \leq r^2 = \delta \sum_{\ell=1}^d \lambda_\ell^2.
\end{align*}
In other words,
\begin{align} \label{eq_sum_of_probabilities_3}
\sum_{\ell=1}^d \lambda_\ell^2 \geq \frac{1}{16\eps \delta  \alpha^2} \log(P(B(W, 2\alpha r), \|\cdot\|_F, 2\eta r))) - \frac{d}{16 \delta  \alpha^2} 
\end{align}
Inequality \ref{eq_sum_of_probabilities_3} contradicts the theorem statement.
Thus, our assumption that   $$\mathbb{P}\left(\langle M_i, \mathcal{A}(M_i) \rangle \geq (1-\eta^2 \delta)\sum_{\ell=1}^d \lambda_\ell^2\right) \geq \frac{1}{2}$$ for every $i \in [n]$ is false, and we therefore have that for some $i \in [n]$ the utility for the matrix $M_i$ satisfies
$$\langle M_i, \mathcal{A}(M_i) \rangle < (1-\eta^2 \delta)\sum_{\ell=1}^d \lambda_\ell^2$$
with probability at least $\frac{1}{2}$.

\end{proof}

\begin{proof}{\bf (of Theorem \ref{thm_lower_bound_general_orbit})}
Consider any $1\leq i<j \leq d$. 
Set $\omega, \zeta, \alpha, \eta, r, \delta$ as follows
\begin{equation}\label{eq_zeta_choice}
    \zeta = \min(\omega, \lambda_1 \sqrt{i}, \lambda_1 \sqrt{d-j+1})\times \frac{(\lambda_i-\lambda_j)}{4C\lambda_1},
    \end{equation}
and $\alpha = \frac{\omega}{2r}$,  $\eta = \frac{\zeta}{2r}$, and
$r= \sqrt{\delta \sum_{\ell=1}^d \lambda_\ell^2}$.
 Then we have 
$\delta = \frac{r^2}{\sum_{\ell=1}^d \lambda_\ell^2} = \frac{\zeta^2}{4 \eta^2 \sum_{\ell=1}^d \lambda_\ell^2}$
Thus, by Lemma \ref{lemma_orbit_packing_improved}, we have that for any $W \in \mathcal{O}_\Lambda$,

\begin{align*} \label{eq_LB_0}
     &\frac{1}{16\eps \delta  \alpha^2} \log(P(B(W, 2\alpha r) \cap \mathcal{O}_\Lambda, \|\cdot\|_F, 2\eta r))) - \frac{d}{16 \delta  \alpha^2}\\
     & \geq  \frac{\sum_{\ell=1}^d \lambda_\ell^2}{4\eps \omega^2} \log(P(B(W, \omega) \cap \mathcal{O}_\Lambda, \|\cdot\|_F, \zeta))) - \frac{d}{\omega^2} \sum_{\ell=1}^d \lambda_\ell^2\\
     & \geq \frac{\sum_{\ell=1}^d \lambda_\ell^2}{4\eps \omega^2} i\times(d-j+1)\log(2) - \frac{d}{\omega^2} \sum_{\ell=1}^d \lambda_\ell^2 \yesnum
     \end{align*}
Consider the following equation for any $\delta>0$: 
 \begin{equation} \label{eq_LB_1}
     \sum_{\ell=1}^d \lambda_\ell^2 < \frac{1}{16\eps \delta  \alpha^2} \log(P(B(W, 2\alpha r) \cap \mathcal{O}_\Lambda, \|\cdot\|_F, 2\eta r))) - \frac{d}{16 \delta  \alpha^2}.
 \end{equation}
By plugging \eqref{eq_LB_0} into \eqref{eq_LB_1}, we get that \eqref{eq_LB_1} holds for any $\omega>0$ such that
\begin{align}\label{eq_LB_2}
     &  \sum_{\ell=1}^d \lambda_\ell^2  < \frac{ \sum_{\ell=1}^d \lambda_\ell^2 }{4\eps \omega^2} i\times(d-j+1)\log(2) - \frac{d}{\omega^2} \sum_{\ell=1}^d \lambda_\ell^2 
     \end{align}
Rearranging \eqref{eq_LB_2} we get
\begin{align*}
     & \omega^2 < \frac{1}{4\eps} i\times(d-j+1)\log(2) - d.
     \end{align*}
     Thus, by Lemma \ref{thm_lower_bound_orbit} we have that for any $\epsilon$-differentially private algorithm $\mathcal{A}$ which takes as input a Hermitian matrix and outputs a matrix in the orbit $\mathcal{O}_\Lambda$,  there exists a Hermitian matrix $M$ with eigenvalues $\lambda_1 \geq \cdots \geq \lambda_d$ such that, with probability at least $\frac{1}{2}$, the output  $\mathcal{A}(M)$ of the algorithm satisfies

\begin{align} \label{eq_LB_3}
  \sum_{\ell=1}^d \lambda_\ell^2 -  \langle M, \mathcal{A}(M) \rangle &\geq \eta^2\delta \sum_{\ell=1}^d \lambda_\ell^2 \notag \\
  &= \frac{\zeta^2}{4}  \notag\\
  &\stackrel{\textrm{Eq. \ref{eq_zeta_choice}}}{=} \min(\omega^2, \lambda_1^2 i, \lambda_1^2 (d-j+1))\times \frac{(\lambda_i-\lambda_j)^2}{64C^2\lambda_1^2}  \notag\\
&=\min\left(\frac{1}{4\eps \lambda_1^2} i\times(d-j+1)\log(2) - \frac{d}{\lambda_1^2}, \,\, i, \,\, d-j+1\right)\times \frac{(\lambda_i-\lambda_j)^2}{64C^2}.
  \end{align}
Since \eqref{eq_LB_3} holds for any $1\leq i < j \leq d$, we must have that any $\epsilon$-differentially private algorithm $\mathcal{A}$ which takes as input a Hermitian matrix and outputs a matrix in the orbit $\mathcal{O}_\Lambda$,  there exists a Hermitian matrix $M$ with eigenvalues $\lambda_1 \geq \cdots \geq \lambda_d$ such that,  with probability at least $\frac{1}{2}$, the output  $\mathcal{A}(M)$ of the algorithm satisfies

\begin{equation} \label{eq_LB_3b}
  \sum_{\ell=1}^d \lambda_\ell^2 -  \langle M \mathcal{A}(M)  \geq \max_{1\leq i < j \leq d} \frac{(\lambda_i-\lambda_j)^2}{64C^2} \times \min\left(\frac{1}{4\eps \lambda_1^2} i\times(d-j+1)\log(2) - \frac{d}{\lambda_1^2}, \,\, i, \,\, d-j+1\right).
  \end{equation}
  Plugging Lemma \ref{lemma_frobenius_inner_produc} into \eqref{eq_LB_3b}, we get
\begin{equation}\label{eq_LB_4}
  \|M- \mathcal{A}(M)\|_F^2  \geq \max_{1\leq i < j \leq d} \frac{(\lambda_i-\lambda_j)^2}{64C^2} \times \min\left(\frac{1}{4\eps \lambda_1^2} i\times(d-j+1) -  \frac{d}{\lambda_1^2}, \,\, i, \,\, d-j+1\right)
  \end{equation}
 Taking the maximum over only pairs $(i,j)$ where $1\leq i <j \leq d$ and either  $j =\frac{d}{2}$, or $i = \frac{d}{2}$, and adjusting the universal constant $C$, we get
 \begin{equation*}
  \|M- \mathcal{A}(M)\|_F^2  \geq \max_{1\leq i \leq \frac{d}{2}}  \frac{(\lambda_i-\lambda_{d-1})^2}{C^2} \times \min\left(\frac{1}{\eps \lambda_1^2} i\times d -  \frac{d}{\lambda_1^2}, \,\, i, \,\, \frac{d}{2}\right)
  \end{equation*}
  and hence that
   \begin{equation} \label{eq_LB_5}
  \|M- \mathcal{A}(M)\|_F^2  \geq \Omega \left(\max_{1\leq i \leq \frac{d}{2}}  (\lambda_i-\lambda_{d-1})^2 \times \min\left(\frac{1}{\eps \lambda_1^2} i\times d, \,\, i \right) \right).
  \end{equation}
  Inequality \ref{eq_LB_5} implies that
  \begin{equation} \label{eq_LB_6}
   \| M - H \|_F^2 \geq \Omega\left(  \frac{d}{{\max(\lambda_1 \sqrt{\eps}, \sqrt{d})^2}} 
\max_{1\leq i \leq \frac{d}{2}}   i \times (\lambda_i -\lambda_{d-i +1})^2\right)
\end{equation}
with probability at least $\frac{1}{2}$.

Inequality \ref{eq_LB_6} proves Theorem \ref{thm_lower_bound_general_orbit} when the output is in the unitary orbit $\mathcal{O}_\Lambda$.
The bound for the setting when the output is a rank-$k$ matrix is a special case of Corollary \ref{cor_lower_covariance_estimation}, and we defer the proof of this fact to the proof of Corollary \ref{cor_lower_covariance_estimation}.
\end{proof}

\begin{proof}{(\bf of Corollary \ref{cor_lower_covariance_estimation})}
Define $\Delta := \sqrt{\frac{cd}{{\max(\lambda_1 \sqrt{\eps}, \sqrt{d})^2}} 
\max_{1\leq i \leq \frac{d}{2}}   i \times (\lambda_i -\lambda_{d-i +1})^2}$.
Let $\mathcal{A}$ be any $\eps$-differentially private algorithm which takes as input a Hermitian matrix $M$ and outputs a matrix $H= \mathcal{A}(M)$.
Consider the algorithm $\hat{\mathcal{A}}$ defined by
\begin{equation*}
  \hat{\mathcal{A}}(M) =  \mathrm{argmin}_{Z \in \mathcal{O}_\Lambda} \|Z -\mathcal{A}(M)\|_F.
\end{equation*}
Since $\hat{\mathcal{A}}$ is just a post-processing of the output of the $\eps$-differentially private  algorithm $\mathcal{A}$,  $\hat{\mathcal{A}}$ must also be $\eps$-differentially private.
Thus, by Theorem \ref{thm_lower_bound_general_orbit} there is a matrix $M \in \mathcal{O}_\Lambda$ such that the output $\hat{H}:= \hat{\mathcal{A}}(M)$ of the projection of $H$ onto $\mathcal{O}_\Lambda$ satisfies
\begin{equation}\label{eq_lower_bound_cor_0}
    \|M- \hat{H}\|_F^2 \geq \Delta^2,
\end{equation}
 with probability at least $\frac{1}{2}$.
Let $E$ be the event that \eqref{eq_lower_bound_cor_0} holds.
Then $\mathbb{P}(E) \geq \frac{1}{2}$.

In the remainder of the proof, we suppose that the event $E$ occurs.
Then \eqref{eq_lower_bound_cor_0} holds and we have
\begin{equation} \label{eq_lower_bound_cor_1}
    \|M- \hat{H}\|_F \geq \Delta.
\end{equation}
We consider the following two cases:  $\|\hat{H} -H\|_F \geq \frac{\Delta}{2}$, and $\|\hat{H} -H\|_F < \frac{\Delta}{2}$.
In the first case where $\|\hat{H} -H\|_F \geq \Delta$, we must have that, since $M \in \mathcal{O}_\Lambda$ and $\hat{H} = \mathrm{argmin}_{Z \in \mathcal{O}_\Lambda} \|Z -H\|_F$,
\begin{align}\label{eq_lower_bound_cor_2}
    \|M-H\|_F &\geq \|\hat{H} -H\|_F \notag \\
    &\geq \frac{\Delta}{2}.
\end{align}
Next, we consider the second case where $\|\hat{H} -H\|_F < \frac{\Delta}{2}$:
By \eqref{eq_lower_bound_cor_1} we have that       $\|M- \hat{H}\|_F \geq \Delta$.
Thus, by the triangle inequality we have
\begin{align}\label{eq_lower_bound_cor_3}
    \|M -H\|_F &\geq \|M- \hat{H}\|_F - \|\hat{H} -H\|_F \notag\\
    & \geq \Delta - \frac{\Delta}{2} \notag \\
    & \geq  \frac{\Delta}{2}.
\end{align}
Therefore, from \eqref{eq_lower_bound_cor_2} and \eqref{eq_lower_bound_cor_3}, we have that
\begin{align}
    \|M -H\|_F \geq \frac{\Delta}{2} \notag
\end{align}
and hence that
\begin{align}
    \|M -H\|_F^2 &\geq \frac{\Delta^2}{4}\notag\\
    &=\frac{c d}{{4\max(\lambda_1 \sqrt{\eps}, \sqrt{d})^2}} 
\max_{1\leq i \leq \frac{d}{2}}   i \times (\lambda_i -\lambda_{d-i +1})^2 \notag
\end{align}
whenever the event $E$ occurs.
Thus, since $c$ is a universal constant, we can choose a slightly different universal constant $c$ such that
\begin{align} \label{eq_lower_bound_full_rank}
    \|M -H\|_F^2 &\geq \frac{c d}{{\max(\lambda_1 \sqrt{\eps}, \sqrt{d})^2}} 
\max_{1\leq i \leq \frac{d}{2}}   i \times (\lambda_i -\lambda_{d-i +1})^2,
\end{align}
with probability at least $\frac{1}{2}$ since $\mathbb{P}(E) \geq \frac{1}{2}$.

Note that, since the output of this algorithm $\mathcal{A}$ is allowed to be any matrix (either full rank or restricted to rank-$k$), the lower bound in \eqref{eq_lower_bound_full_rank} applies to both the setting when the output is rank-$k$ for any $k \in [d]$ and when the output is full rank.
Moreover, in the setting where the output has rank $k<d$, we have also have that $\| M - H \|_F^2 \geq \sum_{\ell=k+1}^d \lambda_\ell^2$ with probability $1$.
This fact, together with \eqref{eq_lower_bound_full_rank} imply that
\begin{align} \label{eq_lower_bound_full_rank}
    \|M -H\|_F^2 &\geq \Omega\left(\sum_{\ell=k+1}^d \lambda_\ell^2 + \frac{d}{{4\max(\lambda_1 \sqrt{\eps}, \sqrt{d})^2}} 
\max_{1\leq i \leq \frac{d}{2}}   i \times (\lambda_i -\lambda_{d-i +1})^2 \right). \notag
\end{align}

\end{proof}

\begin{proof}{(\bf of Corollary \ref{cor_lower_general_gamma_lambda})}
Theorem 2.4 guarantees that for any $\eps$-differentially private algorithm $\mathcal{A}'$, with output eigenvalues $\gamma_1 \geq \cdots \geq \gamma_d \geq 0$, there exists an input matrix $M$ with eigenvalues $\gamma_1 \geq \cdots \geq \gamma_d \geq 0$ such that the output $\mathcal{A}'(M)$ of this algorithm (with the same eigenvalues $\gamma_1 \geq \cdots \geq \gamma_d \geq 0$) satisfies 
\begin{equation*}
\|M-\mathcal{A}'(M)\|_F^2 \geq c^2,
\end{equation*}
where $c^2 := \Omega\left(\sum_{\ell=k+1}^d \gamma_\ell^2 +  \frac{ d}{{\max(\gamma_1 \sqrt{\eps}, \sqrt{d}})^2}
\max_{1\leq i \leq \frac{d}{2}}   i \times (\gamma_i -\gamma_{d-i +1})^2\right)$.

Let $\Gamma = \mathrm{diag}(\gamma_1, \ldots, \gamma_d)$ and $\Lambda = \mathrm{diag}(\lambda_1, \ldots, \lambda_d)$. 
Let $M = U \Gamma U^\ast$ be the spectral decomposition of $M$.
And let $\tilde{M} = U \Lambda U^\ast$ be the matrix with eigenvalues $\lambda_i$, $i \in [d]$, and the same eigenvectors as $M$.
Recall from the statement of Corollary \ref{cor_lower_general_gamma_lambda} that $H = \mathcal{A}(M)$ where $\mathcal{A}$  is an $\epsilon$-differentially private algorithm $\mathcal{A}$ which takes as input a Hermitian matrix and outputs a rank-$k$ Hermitian matrix with eigenvalues $\lambda_1, \ldots, \lambda_k$.
Consider the following two cases:
$$ (i) \ \ \; \|M-\tilde{M}\|_F> \|\tilde{M}-H\|_F \; \mbox{ and }  \; \; (ii) \ \  \; \|M-\tilde{M}\|_F \leq \|\tilde{M}-H\|_F.$$
On the one hand, if $(i)$ holds, then we have 
\begin{equation}\label{eq_general_1}
\|M-H\|_F \geq \|M-\tilde{M}\|_F > \|\tilde{M}-H\|_F \geq c,
\end{equation}
where the first inequality holds since $\tilde{M} = \mathrm{argmin}_{Z \in \mathcal{O}_\Lambda} \|M-Z\|_F$.

On the other hand, if $(ii)$ holds, we still have that $$\|M-\tilde{M}\|_F \leq \|M-H\|_F$$ since $\tilde{M} = \mathrm{argmin}_{Z \in \mathcal{O}_\Lambda} \|M-Z\|_F$.
Thus 
$$\|\tilde{M}-H\|_F \leq \|M-H\|_F + \|M-\tilde{M}\|_F \leq 2 \|M-H\|_F,$$ which implies that  
\begin{equation}\label{eq_general_2}
    \|M-H\|_F \geq \frac{1}{2} \|\tilde{M}-H\|_F \geq \frac{1}{2}c.
    \end{equation}
    Thus, \eqref{eq_general_1} and \eqref{eq_general_2} together imply that
    \begin{equation*} 
\textstyle    \| M - H \|_F^2 \geq \Omega\left(\sum_{\ell=k+1}^d \gamma_\ell^2 +  \frac{ d}{{\max(\gamma_1 \sqrt{\eps}, \sqrt{d}})^2}
\max_{1\leq i \leq \frac{d}{2}}   i \times (\gamma_i -\gamma_{d-i +1})^2\right).
\end{equation*}
\end{proof}

 \section*{Acknowledgments} This research was conducted in part while NKV was a visiting researcher at Google.  
 OM was supported in part by an NSF grant  (CCF-2104528).

\bibliographystyle{plain}
\bibliography{references}

\end{document}